\documentclass[12pt, draftclsnofoot, onecolumn]{IEEEtran}
\ifCLASSINFOpdf
  \usepackage[pdftex]{graphicx}
  \DeclareGraphicsExtensions{.pdf,.jpeg,.png,.eps}
\else
\fi
\usepackage[cmex10]{amsmath}
\usepackage{amsfonts}

\usepackage{multirow}

\newtheorem{theorem}{Theorem}
\newtheorem{proposition}{Proposition}

\newtheorem{lemma}{Lemma}
\newtheorem{definition}{Definition}

\begin{document}
%
\title{Stochastic Geometry Analysis of Ultra-Dense Networks: Impact of Antenna Height and Performance Limits}
%
%
%

\author{Marcin~Filo,~\IEEEmembership{}
	Chuan~Heng~Foh,~\IEEEmembership{}
        Seiamak~Vahid,~\IEEEmembership{}
        and~Rahim~Tafazolli~\IEEEmembership{} \thanks{This paper was presented in part at the IEEE International Conference on Communications Workshops (ICC Workshops), May 2017.}
}

\maketitle

\begin{abstract}
The concept of Ultra Dense Networks (UDN) is often seen as a key enabler of the next generation mobile networks. However, existing analysis of UDNs, including Stochastic Geometry, has not been able to fully determine the potential gains and limits of densification. In this paper we study performance of UDNs in downlink and provide new insights on the impact of antenna height and careful site selection on the network performance. We focus our investigation on the probability of coverage, average cell rate and average area spectral efficiency for networks with regular and random BS deployments. We show that under a path-loss model which considers antenna height there exists an upper limit on network performance which is dependent on the path-loss model parameters. Our analysis shows an interesting finding that even for over-densified networks a non-negligible system performance can be achieved.

\end{abstract}


%
\IEEEpeerreviewmaketitle

\section{Introduction}

\IEEEPARstart{T}{he} concept of Ultra Dense Networks (UDN) is often seen as a key enabler of the next generation mobile networks \cite{6476879, 6845056}. UDNs are built on the concept of Heterogeneous cellular Networks (HetNet) which provide an effective way for increasing network capacity and achieving higher data rates. In contrast to HetNets, UDNs are expected to provide full network coverage and thus they usually consist of significantly higher number of cells densely deployed within a region. Besides, further densification of cells also offers higher area spectral efficiency which can lead to further increase in network capacity. This potential network capacity gain has triggered growing interest in UDNs from research community, industry and standardization bodies.

Investigation of cellular network performance is crucial for proper assessment of different approaches for increasing network capacity. In general, the methods for investigation of cellular network performance (including UDN) can be subdivided into two main groups. The first group consists of approaches based on analytical modeling with simplified assumptions to maintain tractability whilst the second group attempts to capture system operation in detail by performing system-level simulations. Although in general the simulation based methods provide more accurate results (since they can capture the detail of system operation), the investigation based on the analytical models allow for better understanding of network operation and can provide insightful design guidelines which are often difficult to obtain from simulations. However, due to the complexity of the wireless network behavior, capturing its characteristics by means of analytical modeling is a challenging task. Recently, Stochastic Geometry (SG) has emerged as one of the most prominent approaches as it captures spatial characteristics of wireless networks for performance analysis. This feature is especially important for investigation of UDN since its performance behavior is highly dependent on spatial information of the network. This motivated us to use SG as our main tool for the analysis of UDN performance.

Despite the growing interest and effort from research community to further our understanding of UDN performance behavior, several key questions still remain unanswered. One of the fundamental questions is related to the impact of densification on the UDN performance (e.g. area spectral efficiency, coverage)\footnote{See \cite{ding2017performance2} for other examples}. Whilst efforts have been made, conclusive demonstration on how system performance changes with Base Station (BS) density and what happens when BS density tends to infinity remains missing. Another important question is related to the impact of network deployment (e.g. network site planning, antenna height) on the system performance. For instance, although it is a common belief that site planning is not suitable in the context of UDNs due to high cost and complexity, the impact of site planning on UDN performance have not been quantified. In the following, we review key developments of UDN performance investigation.

\subsection{Related work}

One of the first works investigating the relationship between network density and system performance is due to Andrews et al. \cite{andrews2011tractable}. The work shows surprising results where coverage probability as well as mean achievable data rate per cell, do not depend on BS density (also known as SINR invariance property \cite{zhang2015downlink}). The constant mean achievable data rate per cell implies that continuingly increasing the number of BSs in wireless networks could lead to limitless overall network performance improvements. Their conclusions are based on the assumptions of simple unbounded path-loss model, noise-less networks and BS locations following a Poisson point process. Under these assumptions, as BS density increases, the change in aggregated interference power is counter-balanced by the change in signal power, and thus the SINR remains unchanged regardless of network density \cite{andrews2011tractable}. Following the same assumptions, Dhillon et al. \cite{dhillon2012modeling} shows SINR invariance property for HetNets.

The applicability of the above conclusions has been recently questioned in \cite{zhang2015downlink, gupta2015sinr, andrews2015we} where they show that under a multi-slope path-loss model the coverage probability and mean data rate per cell are dependent on BS density. This is in contrast to the earlier SINR invariance property that suggested potentially infinite aggregated data rate of the network resulting from BS densification. Moreover, it has been shown that for certain critical near-field path-loss exponents an optimal network density exists to maximize the coverage probability and mean data rate per cell. Beyond this optimal network density, the coverage probability decreases as BS density increases. Similar results can be also found in \cite{nguyen2016coverage}, where Nguyen and Kountouris also show that SINR invariance does not hold for a multi-slope path-loss model. These works indicate that the path-loss model is a critical factor affecting the network performance of a UDN.

Further studies have also demonstrated that other features of path-loss model and network deployment may also invalidate SINR invariance property in UDN. In the aspect of path-loss model, UDN performance was analysed with various path-loss models in the literture. For unbounded path-loss model including Line Of Sight (LOS) and Non Line of Sight (NLOS) consideration, Ding et al. show in simulation and analytically in \cite{ding2016performance} that the network performance depends on BS density. In \cite{gupta2015potential}, Gupta et al. extend their earlier study of multi-slope path-loss model to include a scenario in which BSs and users are located in 3-dimensional space. They came to a similar conclusion that near-field path-loss exponent influences the performance of the network and SINR invariance property may only be valid under certain conditions. Most recently, Ding and Perez studied the impact of antenna height in UDN with results claiming that the area spectral efficiency goes to zero as BS density goes to infinity \cite{7842150,ding2017performance}. Their results also dismiss the SINR invariance property for the considered scenario.Similar to Ding and Perez, in our previous work we studied the impact of antenna height on the network performance. In contrast to their work, our results hinted however that area spectral efficiency does not necessarily need to go to zero \cite{7962695}. In this paper we provide, among others, additional evidence supporting this claim. 

In the aspect of network deployment, a comparative simulation-based study of irregular and regular BS deployments has been presented in \cite{chen2012small} where Chen et al. show that the area network performance peaks at the certain BS density and then starts to decline showing that network performance in their considered scenarios does depend on BS density for the deployment cases considered. They also show an interesting finding that for certain network densities, the difference in performance between irregular and regular BS deployments is constant approximately.

\subsection{Contributions}
Despite providing new and interesting insights, there are still a lot of questions related to UDN performance which need to be answered. As discussed, SINR invariance property in wireless networks allows Area Spectral Efficiency (ASE) to continue to improve as BS density increases. This is a desirable feature for performance improvement in UDN since it maintains densification gain as network densifies. While SINR invariance property has been shown to appear in some scenarios and system models, as discussed above, recent studies show SINR invariance property does not hold in many other scenarios. In this paper, we focus on the impact of antenna height and placement on UDN performance. In particular, we attempt to answer the following questions:

\begin{itemize}
\item What is the role of antenna height on the SINR invariance property in UDN? Existing papers \cite{ding2016performance,nguyen2016coverage, zhang2015downlink, gupta2015sinr, andrews2015we} study the impact of multi-slope feature as well as LOS/NLOS feature in a path-loss model, and they show that SINR invariance property does not hold with these features. In \cite{7842150}, Ding and Perez consider BS antenna height in the path-loss model and show the invalidation of SINR invariance property for the scenario. In this paper, similar to our earlier work in \cite{7962695}, we not only show the invalidation of SINR invariance property for any nonzero BS antenna height, but also formulate the relationship showing how UDN performance is affected by the BS antenna height. Interestingly, we found conditions where SINR invariance property does hold for the considered path-loss model. We call this \textit{density countering condition} for maintaining the desirable SINR invariance property.

\item For our scenario, what will happen to Area Spectral Efficiency (ASE, in bps/Hz/m$^2$) when BS density approaches infinity? In the case where SINR invariance does not hold, it is important to understand whether continuing densification will still lead to ASE improvement. A recent work claims that ASE goes to zero when BS density approaches infinity \cite{7842150, ding2017performance}. In this paper, we show that ASE does not necessarily go to zero when BS density approaches infinity. We have presented new results to calculate the ASE.

\item Can we benefit from careful site planning in UDN? A recent paper \cite{chen2012small} shows with simulation that irregular and regular BS deployments do give difference ASE in UDN for some network densities. In this paper, we show analytically the performance of regular networks, and investigates the conditions in which regular networks can outperform irregular networks.
\end{itemize}

\section{System Model}

\subsection{Network model}

We consider a single-tier cellular network utilizing a multiple access technique which ensures orthogonal resource allocation within a cell. All BSs in the network transmit with the same power. We further assume that a mobile user always connects to the BS with the strongest received signal which is usually the closest BS to the mobile user. The mobile user density $\lambda_u$ is assumed to be much higher than the BS density $\lambda$ (i.e. $\lambda_u>> \lambda)$ such that each BS always has a user to serve. Regular and irregular BS deployment models for one-dimensional (1D) and two-dimensional (2D) Euclidean space are considered in this work.

In network of regular BS deployment, BSs are arranged in a regular geometrical structure within the network. In the case of 1D regular deployment, BSs are deployed regularly every fixed distance along a line. In the case of 2D, BSs are deployed in a regular hexagonal layout with each BS placed at the center of a hexagon. In other words, the 2D locations of BSs are
\begin{equation*} \Phi^{HEX}=\{ (\Upsilon(m + n/2), \Upsilon(n \sqrt{3}/2))|\, m, n \in \mathbb{Z} \} \in \mathbb{R}^2\end{equation*}

and the 1D locations of BSs are
\begin{equation*} \Phi^{LINE} = \{ \Upsilon i|\, i \in \mathbb{Z} \} \in \mathbb{R}\end{equation*} 
where $\mathbb{Z}$ and $\mathbb{R}$ are sets of integer and real numbers, respectively. The quantity $\Upsilon$ is the inter-site distance between two adjacent BSs. The spatial density of BSs in the regular deployment is $\lambda = \frac{2}{\Upsilon^2 \sqrt{3}}$ and $\lambda = 1/\Upsilon$ for 2D and 1D cases respectively. Mobile users are located uniformly in the network. An illustration of both deployments is given in Figure \ref{fig:system_model}.

In the irregular deployment, the locations of BSs are modelled according to a homogeneous Poisson point process (PPP) $\Phi^{PPP}$ with a spatial density of $\lambda$. The mean inter-site distances for the 2D and 1D cases are $\Upsilon = \sqrt{\frac{2}{\lambda \sqrt{3}}}$ and $\Upsilon=1/\lambda$, respectively. Mobile users are uniformly distributed in the Voronoi cell of its corresponding BS.

\begin{figure}[!t]
\centering
\includegraphics[scale=0.23]{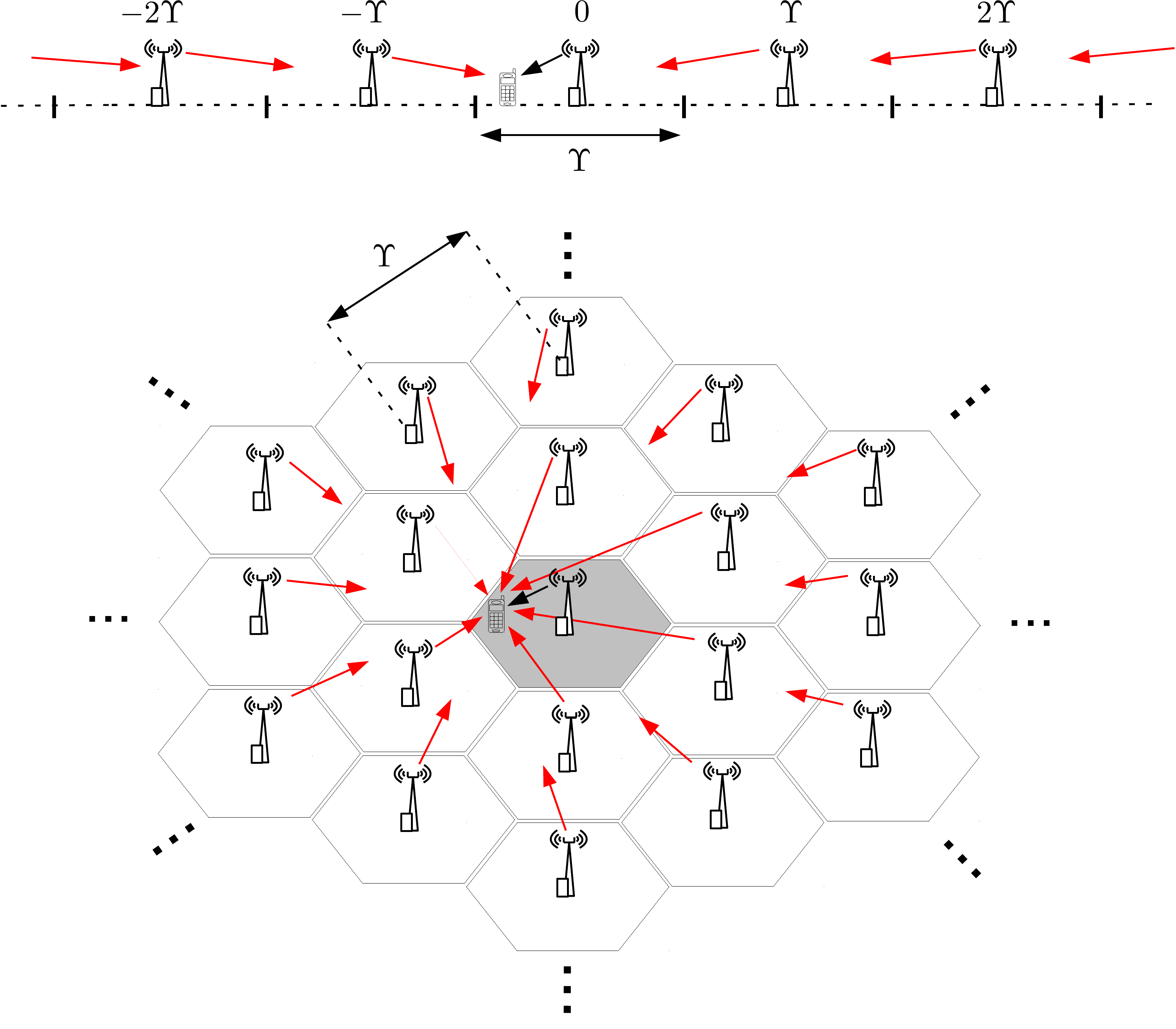}
\caption{Regular BS deployment for 1D (top) and 2D (bottom) Euclidean space}
\label{fig:system_model}
\end{figure}

\subsection{Channel model}

As discussed in the literature, the channel model, particularly the path-loss plays an important role in the performance of the UDN. In this work, we consider a path-loss model function that can describe antenna height, $l_1(r) = (h^2 + r^2)^{-\alpha/2}$, where the parameter $h$ is the difference between BS and mobile user antenna heights, $r$ is the horizontal distance between the mobile user and the BS, and $\alpha$ is the path-loss exponent with $\alpha>1$ for 1D and $\alpha>2$ for 2D. A special case of this path-loss model is $l_0(r) = r^{-\alpha}$ when the antenna height is set to zero. The path-loss model $l_0(r)$ is a commonly used model in network performance analysis (see e.g. \cite{andrews2011tractable} and \cite{dhillon2012modeling}). However, we shall show that $l_0(r)$ produces overly optimistic performance results in UDNs. In this paper, we mainly focus on $l_1(r)$ but also use $l_0(r)$ for comparison.

Unless otherwise stated, the random channel effects are modelled as Rayleigh fading with mean 1, thus the received power (for the desired or interference signal) at a user is $g\cdot l_i(r)$, where $g\sim exp(\mu)$ is an exponentially distributed random variable with a mean value of $1/\mu$ which is equal to the BS transmit power, and $l_i(r)$ is one of our considered two path-loss model functions ($i \in \{0, 1\}$).

It is worth highlighting here that due to high density of BSs in UDNs, desired signal power as well as interference power are significantly higher compared to Gaussian noise. We also show that when the BS density approaches infinity, the UDN with noise in the path-loss model gives exactly the same performance as that without noise (see Theorem \ref{theorem_aase}). This allows us to neglect noise in order to simplify expressions and conduct our investigations.

\section{Coverage Probability Analysis}
\label{sec:cpa}
The coverage probability is a measure of chances that a particular user can be served by the network \cite{andrews2011tractable}. Formally, it is defined as $p_c \overset{\scriptscriptstyle\Delta}{=} \mathbb{P}(SINR > T)$ which is the complementary cumulative distribution function (CCDF) of SINR. For completeness, the results in this section are derived for regular as well as irregular deployments (even though some of the considered cases do not lead to tractable expressions). Based on Slivnyak's theorem, in our considered irregular network where BS locations are modeled as a homogeneous PPP, the coverage probabilities of mobile users are statistically identical \cite{baccelli2009stochastic}. Thus investigating the performance of a typical user in the irregular network is sufficient.

Assuming that a user is always associated with the BS of the strongest received signal, the SINR of a user can be written as
\begin{equation}SINR = \frac{g \cdot l(r)}{I_r + \sigma^2},\end{equation}
where $r$ is the distance between the considered user and its serving BS, $\sigma^2$ is constant receiver noise power, and
\begin{equation}I_r = \sum_{x \in \Phi \setminus b_o} g_xl(\|x-u_0\|),\end{equation}
is the interference power which is the sum of received power from all surrounding BSs except the user's serving BS at $b_o$ and $\|x-u_0\|$ is the distance between each BS and the user located at $u_0$.

Depending on the considered BS deployment the distance $r$ separating a typical user from its serving BS can have different distributions. For irregular BS deployment, as shown in \cite{andrews2011tractable}, $r$ is shown to be Rayleigh distributed. The following expression presents a distribution function of $r$ for a $d$-dimensional Euclidean space 
\begin{equation} f_r(r) = c_d d r^{d-1} \lambda e^{-c_d\lambda r^d} \,, r \geq 0\end{equation}
where $d \in \{1, 2\}$ is the dimension of the Euclidean space and $c_d \overset{\scriptscriptstyle\Delta}{=} \frac{\pi^{d/2}}{\Gamma(1+d/2)}$ is volume of the $d$-dimensional unit ball (see \cite{haenggi2009interference}). Given 1D and 2D space, $c_1 = 2$, $c_2=\pi$.

In case of regular BS deployment, $r$ is uniformly distributed leading to following expressions (for 1D and 2D Euclidean space, respectively)
\begin{equation}f_r(r) = \frac{2}{\Upsilon} \,\,,  r \leq \frac{\Upsilon}{2} \end{equation}
\begin{equation}
\label{eqn_fr_reg_2d}
f_r(r) = \frac{6 \cdot 2 r}{\Upsilon^2\sqrt{3}} \,\,,  0 \leq r \leq \frac{\Upsilon }{2 \sin(\theta + \frac{\pi}{3})}  \end{equation}
where $\Upsilon$ is the inter-site distance and $0 \leq \theta \leq \frac{\pi}{3}$.

\begin{definition}
\label{definition_1} 
Based on Theorem 1 in \cite{andrews2011tractable}, the coverage probability in the downlink for a generic path-loss model function, assuming Rayleigh fading, can be represented as
\begin{equation}
\label{eqn_pc_generic}
p_c^{\{d,i\}} (T, \lambda, \alpha) = \int_{\mathbb{R}^d} e^{-\frac{\mu \sigma^2 T}{l(r)}} \mathcal{L}_{I_r}\left(\frac{\mu T}{l(r)}\right) f_r(r) dr, 
\end{equation}
where $\mathcal{L}_{I_{r}}\left(s\right)$ is the Laplace transform of the random variable $I_r$, $d$ is the dimension of the Euclidean space and $i \in \{0, 1\}$ indicates the path-loss model function. \end{definition}

In the following subsection, we derive formulas for coverage probability for both regular and irregular networks.

\subsection{Irregular network}

We first study the extreme condition where BS density approaches infinity. In the following, we show that the coverage probabilities for irregular 1D and 2D networks under $l_1(r)$ tends to zero when BS density approaches infinity. In case of $l_0(r)$, the probability of coverage is independent of BS density (see also \cite{andrews2011tractable})\footnote{The formula for 2D network is already given in \cite{andrews2011tractable}. Here we provide the formula for 1D network.}. We start by formulating the coverage probability for irregular 1D and 2D networks in the following Lemma, assuming a generic path-loss model.

\begin{lemma}
\label{lemma_1}
The coverage probabilities for irregular 1D and 2D networks under $l_1(r)$ path-loss model function can be expressed as
\begin{equation} 
\label{lemma1_pc_1}
p_c^{\{1,1\}} (T,\lambda, \alpha) = 2\int_0^\infty e^{-2k} e^{-\mu T \sigma^2(\frac{k^2}{\lambda^2}+h^2)^{\alpha/2}} \cdot e^{-2  \xi(T,\alpha,k)} dk
\end{equation}
and

\begin{equation}
\label{lemma1_pc_2}
p_c^{\{2,1\}} (T,\lambda, \alpha) = \pi e^{-\lambda \pi h^2 \rho_2(T,\alpha)} \int_0^\infty e^{-\mu T \sigma^2 (\frac{k}{\lambda}+h^2)^{\alpha/2}} \cdot e^{-\pi k(1+ \rho_2(T,\alpha))} dk
\end{equation}
respectively where
\begin{equation}
\xi(T,\alpha,k) = \int_k^\infty \frac{1}{1 + T^{-1} \left(\frac{t^2 + (h\lambda)^2}{k^2 + (h\lambda)^2}\right)^{\frac{\alpha}{2}}} dt 
\end{equation}
and
\begin{equation} 
\label{eqn_pc_d1_2}
\rho_d(T,\alpha) = T^{d/\alpha}\int_{T^{-d/\alpha}}^\infty   \frac{1}{1+u^{\alpha/d}}  du.
\end{equation}

\end{lemma}

\begin{proof}
The proof is based on the proof of Theorem 1 in \cite{andrews2011tractable}. By substituting $f_r(r) = c_d d r^{d-1} \lambda e^{-c_d\lambda r^d}$ and $l(r) = l_1(r)$ in (\ref{eqn_pc_generic}), the coverage probability of $d$-dimensional irregular network can be expressed as
\begin{equation}
\label{eqn_pc_2d2_proof_1}
p_c^{\{d,1\}} (T,\lambda, \alpha) = dc_d\lambda\int_0^\infty r^{d-1}e^{-c_d\lambda r^d} e^{-\mu T \sigma^2 (r^2+h^2)^{\alpha/2}} \cdot \mathcal{L}_{I_r}\left(\mu T(r^2+h^2)^{\alpha/2}\right) dr
\end{equation} where
\begin{equation} \mathcal{L}_{I_{r}}\left(s\right) = \exp \left( -dc_d\lambda \int_{r}^\infty  \left(  1- \frac{\mu}{\mu +s (h^2+v^2)^{-\frac{\alpha}{2}}}  \right) v^{d-1}dv \right).
\end{equation} 
By assuming $d=2$ and employing a change of variables $u = \frac{v^2+h^2}{T^{\frac{2}{\alpha}}(r^2 + h^2)}$ we obtain the following expression for 2D
\begin{equation} \mathcal{L}_{I_{r}}\left(\mu T(h^2+r^2)^{\alpha/2}\right) = e^{-\lambda \pi (r^2+h^2)\,\rho_2(T,\alpha)}. \end{equation}

When we substitute the above result in (\ref{eqn_pc_2d2_proof_1}) and employ a change of variables $k = \lambda r^2$, we obtain the final result for 2D. 

Similarly, by assuming $d=1$ and employing a change of variables $k = \lambda r$ we obtain following expressions for 1D
\begin{multline}
p_c^{\{d,1\}} (T,\lambda, \alpha) = 2\int_0^\infty e^{-2k} e^{-\mu T \sigma^2 (\frac{k^2}{\lambda^2}+h^2)^{\alpha/2}}  \cdot \exp \left( -2\lambda \int_{\frac{k}{\lambda}}^\infty \frac{1}{1 + T^{-1} \left(\frac{(v\lambda)^2 + (h\lambda)^2}{k^2 + (h\lambda)^2}\right)^{\frac{\alpha}{2}}} dv\right) dr.
\end{multline}

Next, by employing a change of variables $t = \lambda v$, we obtain the final result for 1D, thus concluding the proof.

\end{proof}

Based on Lemma \ref{lemma_1}, we formulate the following theorem.

\begin{theorem} 
\label{theorem_1}
The coverage probability of 1D and 2D irregular networks under path-loss model $l_1(r)$ tends to 0 for BS density $\lambda \to \infty$, that is
\begin{equation} \lim_{\lambda \to \infty} p_c^{\{d,1\}} (T,\lambda, \alpha) =  0. \end{equation}  
\end{theorem}
\begin{proof} 
By applying the Monotone Convergence Theorem to formula (\ref{lemma1_pc_1}) for 1D network, and by using basic properties of limits, we can rewrite (\ref{lemma1_pc_1}) as 
\begin{multline} 
\lim_{\lambda \to \infty} p_c^{\{1,1\}} (T,\lambda, \alpha) = 2\int_0^\infty e^{-2k} \lim_{\lambda \to \infty} e^{-\mu T \sigma^2(\frac{k^2}{\lambda^2}+h^2)^{\alpha/2}} \\ \cdot \exp \left(-2  \int_k^\infty \lim_{\lambda \to \infty} \frac{1}{1 + T^{-1} \left(\frac{t^2 + (h\lambda)^2}{k^2 + (h\lambda)^2}\right)^{\frac{\alpha}{2}}} dt\right) dk.
\end{multline}

Simplifying the above, yields 
\begin{equation} 
\lim_{\lambda \to \infty} p_c^{\{1,1\}} (T,\lambda, \alpha) = 2\int_0^\infty e^{-2k} e^{-\mu T \sigma^2 h^\alpha}  \cdot \exp \left(-2  \int_k^\infty \frac{1}{1 + T^{-1}} dt\right) dk.
\end{equation}

We finally conclude the proof for 1D network by exploiting the fact that 
\begin{equation} \int_k^\infty \frac{1}{1 + T^{-1}} dt = \infty, \end{equation}
which brings the coverage probability to 0 as $\lambda \to \infty$.

Similar to the 1D case, by applying the Monotone Convergence Theorem to formula (\ref{lemma1_pc_2}) for 2D network, and by using basic properties of limits, we have 
\begin{multline}
\lim_{\lambda \to \infty} p_c^{\{2,1\}} (T,\lambda, \alpha) = \lim_{\lambda \to \infty} \pi e^{-\lambda \pi h^2 \rho_2(T,\alpha)} \cdot \int_0^\infty \lim_{\lambda \to \infty}  e^{-\mu T \sigma^2 (\frac{k}{\lambda}+h^2)^{\alpha/2}} e^{-\pi k(1+ \rho_2(T,\alpha))} dk.
\end{multline}
From the above results, we can see that
\begin{equation} \lim_{\lambda \to \infty} \pi e^{-\lambda \pi h^2 \rho_2(T,\alpha)} = 0, \end{equation}
and
\begin{equation} \lim_{\lambda \to \infty} e^{-\mu T \sigma^2 (\frac{k}{\lambda}+h^2)^{\alpha/2}} = e^{-\mu T \sigma^2 h^\alpha}, \end{equation}
thus concluding the proof for 2D.

\end{proof}

It is worth noting that Theorem \ref{theorem_1} shows the dependence of coverage probability on BS density for the considered path-loss model. This is in line with the recent works \cite{zhang2015downlink, gupta2015sinr, andrews2015we, ding2016performance} which dismiss the SINR invariance property for other path-loss models. In other words,  for $h > 0$, the increase in the interference power is not counter-balanced by the increase in the signal power.

Interestingly, we also observe that changes in the BS density $\lambda$ can be counter-balanced by the adjustments of the path-loss model parameter $h$ to maintain the same coverage probability. We call this the \textit{density countering condition} which leads to maintenance of SINR invariance property. We describe this condition in the following theorem.

\begin{theorem}
\label{theorem_2}
The coverage probability of noise-less irregular networks for path-loss model $l_1(r)$ is constant, when $\lambda\, h^d=c$ and $c$ is constant.
\end{theorem}

\begin{proof}
The proof is based on Lemma \ref{lemma_1}. By substituting $\sigma^2 \to 0$ and $\lambda h^d$ into (\ref{lemma1_pc_1}) for 1D and (\ref{lemma1_pc_2}) for 2D, we obtain the following expressions which are constant for a constant $c$ where
\begin{equation} 
p_c^{\{1,1\}} (T,\lambda, \alpha) = 2\int_0^\infty e^{-2k}  \cdot \exp \left(-2  \int_k^\infty \frac{1}{1 + T^{-1} \left(\frac{t^2 + c^2}{k^2 + c^2}\right)^{\frac{\alpha}{2}}} dt \right) dk
\end{equation}
and
\begin{equation}
p_c^{\{2,1\}} (T,\lambda, \alpha) = \frac{e^{-c \pi \rho_{2}(T,\alpha)} }{1+ \rho_{2}(T,\alpha))}.
\end{equation}

\end{proof}

The coverage probabilities for irregular 1D and 2D networks under $l_1(r)$ can be computed numerically using Lemma \ref{lemma_1}. The results permit simple expressions for noise-less condition. For the case of 2D network, we can derive a closed form expression. In case of 1D network, an expression which requires a single numerical integration can be obtained for some specific $\alpha$ values, such as $\alpha=2$. Both expressions are presented in the following proposition.

\begin{proposition}
\label{proposition_1}
The coverage probability under $l_1(r)$ path-loss model function for noise-less 1D (given $\alpha = 2$) and 2D networks can be expressed as follows
\begin{equation} 
\label{eqn_pc_2d1_first}
p_c^{\{1,1\}} (T,\lambda, 2) = 2\int_0^\infty e^{-2k}  \cdot e^{2 \frac{T((\lambda h)^2 + k^2) (\arctan \left( k^{-1} \sqrt{ T ((\lambda h)^2 + k^2) + (\lambda h)^2 }\right) )   }{ \sqrt{T ((\lambda h)^2 + k^2) + (\lambda h)^2}} }dk
\end{equation}

\begin{equation} 
\label{eqn_pc_2d2_first}
p_c^{\{2,1\}}  (T,\lambda, \alpha) =  \frac{e^{-\lambda \pi h^2 \rho_{2}(T,\alpha)} }{1+ \rho_{2}(T,\alpha))}. \end{equation}

\end{proposition}

\begin{proof}
By assuming $\sigma^2 \to 0$, the proof follows directly from Lemma \ref{lemma_1}.\end{proof}

Notice that formulas for $l_1(r)$ reduce to the coverage probability for $l_0(r)=r^{-\alpha}$ in \cite{andrews2011tractable}, when $h = 0$. For completeness, the formula for noise-less 1D and 2D networks under $l_0(r)=r^{-\alpha}$ is presented below
\begin{equation} p_c^{\{d,0\}}  (T,\lambda, \alpha) =  \frac{1}{1+ \rho_{d}(T,\alpha))}. \end{equation}

\subsection{Regular 1D network}
\label{pc_regular_1D}
We consider a network where the locations of BSs are generated using a deterministic points process $\Phi^{LIN}$. We study the coverage probability of a user located at a particular point in this network. Without loss of generality we assume that this user is located in the origin and thus the BS locations form the following point process on $\mathbb{R}$ 
\begin{equation}\hat{\Phi}^{LIN}_r = \{x - y\,|\, x \in \Phi^{LIN}\},\end{equation}
where $r = \|y\|$ is a distance between a user and its serving (closest) BS at point $y$.

We now formulate two lemmas for the coverage probability in 1D regular network under the path-loss models $l_0(r)$ and $l_1(r)$.

\begin{lemma}
\label{lemma_2}
The coverage probability for 1D regular network under $l_0(r)$ path-loss model function can be expressed as
\begin{equation}
p_c^{\{1,0\}} (T,\lambda, \alpha) = \frac{2}{\Upsilon} \int_0^{\frac{\Upsilon}{2}} e^{-\mu T \sigma^2 r^{\alpha}} \mathcal{L}_{I_r}\left(\mu T r^{\alpha}\right) dr,
\end{equation}
where $\Upsilon = \frac{1}{\lambda}$, and 
\begin{equation} \mathcal{L}_{I_{r}}\left(s\right) = 2\prod_{\substack{i=1}}^{\infty}\left(\frac{1}{1 + (r + i\, \Upsilon)^{-\alpha} \frac{s}{\mu}} \right). \end{equation}

\end{lemma}

\begin{proof}
The proof of Lemma \ref{lemma_2} is based on the proof of Theorem 1 provided in \cite{andrews2011tractable}. As the desired signal is assumed to be exponentially distributed, by substituting $f_r(r) = \frac{2}{\Upsilon}$ in (\ref{eqn_pc_generic}), the coverage probability in 1D regular network under a generic path-loss model function can be expressed as
\begin{equation} 
\label{lemma_2_pc}
p_c^{\{1,i\}} (T,\lambda, \alpha) = \frac{2}{\Upsilon} \int_0^{\frac{\Upsilon}{2}} e^{-\frac{\mu T \sigma^2}{l(r)}} \mathcal{L}_{I_r}\left(\frac{\mu T}{l(r)}\right) dr. 
\end{equation}

Using the definition of the Laplace transform we can show that
\begin{multline} \mathcal{L}_{I_{r}}\left(s\right) = \mathbb{E}_{I_r}[e^{-sI_r}] = \mathbb{E}_{\{g_x\}}\left[\exp(-s\sum_{x \in \hat{\Phi}^{LIN}_r \setminus \{b_o\}} g_x l(\|x\|))\right]  \overset{(a)}{=} \prod_{x \in \hat{\Phi}^{LIN}_r \setminus \{b_o\}} \mathbb{E}_g [\exp (-sg\, l(\|x\|))] \\ \overset{(b)}{=} \prod_{x \in \hat{\Phi}^{LIN}_r \setminus \{b_o\}} \frac{\frac{\mu}{l(\|x\|)}}{\frac{\mu}{l(\|x\|)} + s}\end{multline} 
 
where (a) follows from the i.i.d distribution of $g_x$ and its independence from the point process $\hat{\Phi}^{LIN}_r$ and (b) from the Laplace transform of an exponential random variable with mean $\frac{\mu}{l(\|x\|)}$. Using the above result, and taking into account that locations of BSs in 1D regular network follow ${\Phi}^{LIN}$, the Laplace transform can be further expressed as 
\begin{equation} \mathcal{L}_{I_{r}}\left(s\right) = \prod_{\substack{i=-\infty \\ i \neq 0}}^{\infty}\left(\frac{1}{1 + l(\|r + i\, \Upsilon\|) \frac{s}{\mu}} \right), \end{equation} where $\|r + i\,\Upsilon\|$ is the distance from the typical user to the \textit{i-th} base station.

By substituting $l(r) = l_0(r)$ into (\ref{lemma_2_pc}) and removing the absolute value expression from the product in the expression for the Laplace transform presented above, we immediately obtain the final result.
\end{proof}

\begin{lemma}
\label{lemma_3}
The coverage probability for 1D regular network under $l_1(r)$ path-loss model function can be expressed as
\begin{equation}
   p_c^{\{1,1\}} (T,\lambda, \alpha) = \frac{2}{\Upsilon} \int_0^{\frac{\Upsilon}{2}} e^{-\mu T \sigma^2 (r^2 + h^2)^{\alpha/2}} 
\cdot \mathcal{L}_{I_r}\left(\mu T (r^2 + h^2)^{\alpha/2}\right) dr,
\end{equation}

where $\Upsilon = \frac{1}{\lambda}$, and
\begin{equation} \mathcal{L}_{I_{r}}\left(s\right) = 2\prod_{\substack{i=1}}^{\infty}\left(\frac{1}{1 + ((r + i\, \Upsilon)^2 + h^2)^{-\alpha/2} \frac{s}{\mu}} \right). \end{equation}

\end{lemma}

\begin{proof}
The proof follows that of Lemma \ref{lemma_2}.
\end{proof}

It is worth noting here that setting $h=0$ in $l_1(r)$ reduces $p_c^{\{1,1\}}(T, \lambda, \alpha)$ to $p_c^{\{1,0\}}(T, \lambda, \alpha)$.

Based on Lemma \ref{lemma_2} and Lemma \ref{lemma_3}, we have the following theorems (i.e. Theorem \ref{theorem_3} - \ref{theorem_5}).

\begin{theorem}
\label{theorem_3}
The coverage probability of noise-less 1D regular network under the standard path-loss model function $l_0(r)$ does not depend on BS density $\lambda$.
\end{theorem}

\begin{proof}
The proof is based on Lemma \ref{lemma_2}. By substituting $t = \frac{r}{\Upsilon}$ and assuming $\sigma^2 \to 0$ we obtain
\begin{equation} 
p_c^{\{1,0\}} (T,\lambda, \alpha) = 4 \int_0^{\frac{1}{2}} \prod_{\substack{i=1}}^{\infty}\left(\frac{1}{1 + T\frac{t^{\alpha}}{(t + i)^{\alpha}}} \right) dt.
\end{equation}

From the above formula it can be easily seen that the probability of coverage is independent of $\lambda$, thus concluding the proof.
\end{proof}

\begin{theorem}
\label{theorem_4} 
The coverage probability of noise-less 1D regular networks under path-loss model $l_1(r)$ tends to 0 as BS density $\lambda \to \infty$, that is

\begin{equation} \lim_{\lambda \to \infty} p_c^{\{1,1\}} (T,\lambda, \alpha) =  0. \end{equation}
\end{theorem}

\begin{proof}
The proof is based on Lemma \ref{lemma_3}. By substituting $t = \frac{r}{\Upsilon}$, and given that $\Upsilon = \frac{1}{\lambda}$ we then obtain
\begin{equation}
\label{pc_1_reg}
p_c^{\{1,1\}} (T,\lambda, \alpha) = 4 \int_0^{\frac{1}{2}} e^{-\mu T \sigma^2 ((\frac{t}{\lambda})^2 + h^2)^{\alpha/2}} \cdot \prod_{\substack{i=1}}^{\infty}\left(\frac{1}{1 + T\left(\frac{t^2 + (h\lambda)^2}{(t + i)^2 + (h\lambda)^2}\right)^{\alpha/2}} \right) dt
\end{equation}

By taking $\lambda \to \infty$ and applying the Monotone Convergence Theorem to above formula, we further get 
\begin{equation} 
\lim_{\lambda \to \infty} p_c^{\{1,1\}} (T,\lambda, \alpha) = 4 \int_0^{\frac{1}{2}} \lim_{\lambda \to \infty} e^{-\mu T \sigma^2 ((\frac{t}{\lambda})^2 + h^2)^{\alpha/2}}  \cdot \prod_{\substack{i=1}}^{\infty} \lim_{\lambda \to \infty} \left(\frac{1}{1 + T\left(\frac{t^2 + (h\lambda)^2}{(t + i)^2 + (h\lambda)^2}\right)^{\alpha/2}} \right) dt.
\end{equation}

Simplifying the above expression yields 
\begin{equation}
\label{interm_t4_proof} 
\lim_{\lambda \to \infty} p_c^{\{1,1\}} (T,\lambda, \alpha) = 4 \int_0^{\frac{1}{2}} e^{-\mu T \sigma^2 h^{\alpha}} \prod_{\substack{i=1}}^{\infty} \left(\frac{1}{1 + T} \right) dt.
\end{equation}

Next, given that $T > 0$, 
\begin{equation} \prod_{\substack{i=1}}^{\infty} \frac{1}{1 + T} = 0, \end{equation} 
which brings the expression in (\ref{interm_t4_proof}) to zero.

\end{proof}

Theorems \ref{theorem_3} and \ref{theorem_4} show that similar to 1D irregular network, the 1D regular network does not exhibit SINR invariance property for $l_1(r)$ path-loss model. However, it does exhibit SINR invariance property (similar to noise-less irregular networks) for $l_0(r)$ path-loss model.

In the previous subsection, we showed the density countering condition for irregular 1D and 2D networks under $l_1(r)$. As shown in the following theorem, this condition also applies to 1D regular network.

\begin{theorem}
\label{theorem_5}
The coverage probability of noise-less 1D regular networks under path-loss model $l_1(r)$ is constant when $\lambda\, h = c$ and $c$ is constant.
\end{theorem}

\begin{proof}
The proof is based on the proof of Theorem \ref{theorem_4}. By substituting $\sigma^2 = 0$ and $\lambda h = c$ in (\ref{pc_1_reg}) we obtain the following expression which is constant for a constant $c$
\begin{equation}
p_c^{\{1,1\}} (T,\lambda, \alpha) = 4 \int_0^{\frac{1}{2}} \prod_{\substack{i=1}}^{\infty}\left(\frac{1}{1 + T\left(\frac{t^2 + c^2}{(t + i)^2 + c^2}\right)^{\alpha/2}} \right) dt.
\end{equation}
\end{proof}

Lemma \ref{lemma_2} allows us to numerically calculate coverage probability for arbitrary $\alpha > 1$ and a generic path-loss model function $l(r)$. However, by considering some integer $\alpha$ values we can obtain a simpler expression which allow us to gain additional insight. In the following propositions, we derive coverage probability expressions for $\alpha = 2$ for the considered path-loss model functions.

\begin{proposition}
\label{proposition_1a}
The coverage probability for 1D regular network using the standard path-loss model function $l_0(r) = r^{-\alpha}$, when $\alpha = 2$ is
\begin{equation} 
\label{theorem_3_eq} 
p_c^{\{1,0\}} (T, \lambda, 2) =  \frac{(1+T)}{\pi} \int_0^{\pi} e^{-\frac{\mu T \sigma^2 x^2}{(2\pi\lambda)^2}} \cdot \frac{\cos(x)-1}{\cos(x)-\cosh(x\sqrt{T})} dx. \end{equation}
\end{proposition}

\begin{proof}
Using Lemma \ref{lemma_2}, we start by expressing the Laplace transform of $I_{r}$ as
\begin{equation} \mathcal{L}_{I_{r}}\left(s\right) = \prod_{\substack{i=-\infty \\ i \neq 0}}^{\infty}\left(\ \frac{1}{ 1 + \frac{\frac{s}{\mu \Upsilon^2}}{(\frac{r}{\Upsilon} + i)^2} }\right).\end{equation} 

By using the following expression for an infinite product
\begin{equation} 
\label{inf_product}
\prod_{k=-\infty}^{\infty}\left( 1+ \frac{z}{(k+a)^2+ b^2 }\right) = \frac{\cos(2\pi a) - \cos(2\pi\sqrt{-b^2-z})}{\cos(2\pi a) - \cosh(2\pi b)},
\end{equation} 
and with the setting of $b=0$, $z = \frac{-s}{\mu \Upsilon^2}$, $a=\frac{r}{\Upsilon}$ and the exclusion of $k=0$,we obtain
\begin{equation} \mathcal{L}_{I_{r}}\left(s\right) = \frac{(r^2 + \frac{s}{\mu})(\cos(2\pi\frac{r}{\Upsilon})-1)}{r^2(\cos(2\pi\frac{r}{\Upsilon})-\cosh(2\pi\frac{r}{\Upsilon}\sqrt{T}))}. \end{equation} 

Plugging in $s=\mu T r^2$ and substituting $x=2\pi\lambda r$ to the above result concludes the proof.
\end{proof}

\begin{proposition} 
\label{proposition_2a}
The coverage probability for 1D regular network using the path-loss model function $l_1(r)$, when $\alpha = 2$ can be expressed as
\begin{multline} p_c^{\{1,1\}} (T, \lambda, 2) =  \frac{1+T}{\pi} \int_0^\pi e^{-\frac{\mu T \sigma^2 x^2}{(2\pi\lambda)^2}} \cdot \frac{\cos(x)-\cosh(2\pi\lambda h)}{\cos(x)-\cosh(\sqrt{(2\pi\lambda h)^2 + T (x^2 + (2\pi\lambda h)^2)})} dx. \end{multline}
\end{proposition}

\begin{proof} 
The proof follows that of Proposition \ref{proposition_1a} but in (\ref{inf_product}), the substitution $b=\frac{h}{d}$ is used instead.
\end{proof}

\subsection{Regular 2D network}

Similar to the 1D regular network, we consider a network generated by the deterministic point process $\Phi^{HEX}$, and we study the coverage probability of a user located at the origin. We first see that the BS locations form the following point process on $\mathbb{R}^2$
\begin{equation}\hat{\Phi}^{HEX}_{r,\theta} = \{(x - r\cos\theta, y - rsin\theta)\,|\, (x,y) \in \Phi^{HEX} \},\end{equation}
where $r$ and $\theta$ are the distance and angle between the user, and its serving BS, respectively.

Following a similar approach as in 1D regular network, we first provide the following lemma for the coverage probability in 2D regular network, assuming a generic path-loss model.

\begin{lemma}
\label{lemma_4}
The coverage probability for 2D regular network under $l_0(r)$ path-loss model function can be expressed as
\begin{equation} 
\label{lemma_4_pc}
p_c^{\{2,0\}} (T, \lambda, \alpha) = \frac{12}{\Upsilon^2\sqrt{3}} \int_0^{\frac{\pi}{3}}\int_0^{\frac{\Upsilon}{2 \sin(\theta + \frac{\pi}{3})}} e^{-\mu \sigma^2 T r^{\alpha}} \cdot \mathcal{L}_{I_{r,\theta}}\left(\mu Tr^{\alpha}\right) r dr\, d\theta, 
\end{equation} where $\Upsilon = \sqrt{\frac{2}{\lambda \sqrt{3}}}$, and

\begin{multline} \mathcal{L}_{I_{r,\theta}}\left(s\right) = \prod_{\substack{(n,m) \in \\ \mathbb{Z}^2 \setminus \{(0,0)\}}}\Bigg[1 + \bigg(\Big(\Upsilon(m+\frac{n}{2}) - r\cos\theta\Big)^2 + \Big(\Upsilon n\frac{\sqrt{3}}{2} - r\sin\theta\Big)^2\bigg)^{-\frac{\alpha}{2}} \frac{s}{\mu} \Bigg]^{-1}. \end{multline}

\end{lemma}

\begin{proof}
The proof of Lemma \ref{lemma_4} is analogous to that of Lemma \ref{lemma_2} and is based on the proof of Theorem 1 provided in \cite{andrews2011tractable}. As the desired signal is assumed to be exponentially distributed, by substituting (\ref{eqn_fr_reg_2d}) in (\ref{eqn_pc_generic}), the coverage probability can be expressed as

\begin{equation}
\label{eqn_pc_generic_2d}
p_c^{\{2,i\}} (T, \lambda, \alpha) = \frac{12}{\Upsilon^2\sqrt{3}} \int_0^{\frac{\pi}{3}}\int_0^{\frac{\Upsilon}{2 \sin(\theta + \frac{\pi}{3})}} e^{-\frac{\mu \sigma^2 T}{l(r)}} \cdot \mathcal{L}_{I_{r,\theta}}\left(\frac{\mu T}{l(r)}\right) r dr\, d\theta. 
\end{equation}

Similar to Lemma \ref{lemma_1}, by using the definition of the Laplace transform we can show that

\begin{multline}
   \mathcal{L}_{I_{r,\theta}}\left(s\right) = \mathbb{E}_{I_{r,\theta}}[e^{-sI_{r,\theta}}] = \mathbb{E}_{\{g_u\}}\left[\exp(-s\sum_{u \in \hat{\Phi}^{HEX}_{r,\theta} \setminus \{b_o\}} g_u l(\|u\|))\right] \overset{(a)}{=} \prod_{u \in \hat{\Phi}^{HEX}_{r,\theta} \setminus \{b_o\}} \mathbb{E}_g [\exp (-sg\, l(\|u\|))] \\ 
   \overset{(b)}{=} \prod_{u \in \hat{\Phi}^{HEX}_{r,\theta} \setminus \{b_o\}} \frac{\frac{\mu}{l(\|u\|)}}{\frac{\mu}{l(\|u\|)} + s}
\end{multline}
 
where (a) follows from the i.i.d distribution of $g_u$ and it is independence of the point process $\hat{\Phi}^{HEX}_{r,\theta}$, and (b) follows from the Laplace transform of an exponential random variable with mean $\frac{\mu}{l(\|u\|)}$. Using the above, and taking into account that locations of BSs in 2D regular network follow ${\Phi}^{HEX}$, the Laplace transform can be further expressed as 
\begin{equation*} \mathcal{L}_{I_{r,\theta}}\left(s\right) = \prod_{\substack{(n,m) \in \\ \mathbb{Z}^2 \setminus \{(0,0)\}}}\Bigg[1 + l(\|\Upsilon(m+\frac{n}{2}) - r\cos\theta, \Upsilon n\frac{\sqrt{3}}{2} - rsin\theta)\|) \frac{s}{\mu}\Bigg]^{-1}. \end{equation*}

By substituting $l(r) = l_0(r)$ into (\ref{eqn_pc_generic_2d}), and removing the absolute value expression from the product in the expression for the Laplace transform presented above, we immediately obtain the final result which concludes the proof.
\end{proof}

\begin{lemma}
\label{lemma_5}
The coverage probability for 2D regular network under $l_1(r)$ path-loss model function can be expressed as
\begin{multline}
\label{lemma_5_pc} 
p_c^{\{2,1\}} (T, \lambda, \alpha) = \frac{12}{\Upsilon^2\sqrt{3}} \int_0^{\frac{\pi}{3}}\int_0^{\frac{\Upsilon}{2 \sin(\theta + \frac{\pi}{3})}} e^{-\mu \sigma^2 T (r^2 + h^2)^{\alpha/2}} \cdot \mathcal{L}_{I_{r,\theta}}\left(\mu T(r^2 + h^2)^{\alpha/2}\right) r dr\, d\theta, 
\end{multline} where 
\begin{multline} \mathcal{L}_{I_{r,\theta}}\left(s\right) = \prod_{\substack{(n,m) \in \\ \mathbb{Z}^2 \setminus \{(0,0)\}}}\Bigg[1 + \bigg(\Big(\Upsilon(m+\frac{n}{2}) - r\cos\theta\Big)^2 + \Big(\Upsilon n\frac{\sqrt{3}}{2} - r\sin\theta\Big)^2 + h^2\bigg)^{-\frac{\alpha}{2}} \frac{s}{\mu} \Bigg]^{-1}. \end{multline}

\end{lemma}

\begin{proof}
The proof follows that of Lemma \ref{lemma_4}.
\end{proof}

\begin{theorem}
\label{theorem_6}
The coverage probability for noise-less 2D regular network under the standard path-loss model function $l_0(r) = r^{-\alpha}$ does not depend on BS density $\lambda$.

\end{theorem}

\begin{proof}
The proof is based on Lemma \ref{lemma_4}. By substituting $t = \frac{r^2}{\Upsilon^2}$ and assuming $\sigma^2 \to 0$, the coverage probability of 2D regular network can be expressed as
\begin{multline*} p_c^{\{2,0\}} (T, \lambda, \alpha) = 2 \sqrt{3} \int_0^{\frac{\pi}{3}}\int_0^{\frac{1}{2 \sin(\theta + \frac{\pi}{3})}} \prod_{\substack{(n,m) \in \\ \mathbb{Z}^2 \setminus \{(0,0)\}}}\left(\frac{1}{1 + \frac{T \,t^{\frac{\alpha}{2}}}{\left((m+\frac{n}{2} - \sqrt{t}\cos\theta)^2 + (n\frac{\sqrt{3}}{2} - \sqrt{t}\sin\theta)^2\right)^{\frac{\alpha}{2}}}} \right) dt d\theta. \end{multline*}

From the above formula it can be easily seen that the coverage probability is independent of $\lambda$.
\end{proof}

\begin{theorem}
\label{theorem_7} 
The coverage probability of 2D regular networks under the path-loss model $l_1(r)$ tends to 0 as BS density $\lambda \to \infty$, that is
\begin{equation} \lim_{\lambda \to \infty} p_c^{\{2,1\}} (T,\lambda, \alpha) =  0. \end{equation}
\end{theorem}

\begin{proof}
The proof is based on Lemma \ref{lemma_5}. We start by substituting $t = \frac{r^2}{\Upsilon^2}$ and $\Upsilon^2 = \frac{2}{\lambda\sqrt{3}}$ into (\ref{lemma_5_pc}). Next, by applying the Monotone Convergence Theorem to the obtained formula and some algebraic manipulation, we obtain the following expression 

\begin{multline} 
\label{limit_expression}
\lim_{\lambda \to \infty} p_c^{\{2,1\}} (T,\lambda, \alpha) = 2 \sqrt{3} \int_0^{\frac{\pi}{3}}\int_0^{\frac{1}{2 \sin(\theta + \frac{\pi}{3})}} \lim_{\lambda \to \infty} e^{-\mu \sigma^2 T (\frac{2}{\lambda\sqrt{3}} t + h^2)^{\alpha/2}} \\ \cdot \prod_{\substack{(n,m) \in \\ \mathbb{Z}^2 \setminus \{(0,0)\}}} \lim_{\lambda \to \infty} \left(\frac{1}{1 + T \left(\frac{t + \frac{\sqrt{3}}{2} \lambda h^2 }{(m+\frac{n}{2} - \sqrt{t}\cos\theta)^2 + (n\frac{\sqrt{3}}{2} - \sqrt{t}\sin\theta)^2 + \frac{\sqrt{3}}{2} \lambda h^2}\right)^{\frac{\alpha}{2}}} \right) dt d\theta. \end{multline}
The above formula can be further simplified to

\begin{equation} 
\label{interm_t7_proof}
\lim_{\lambda \to \infty} p_c^{\{2,1\}} (T,\lambda, \alpha) = 2 \sqrt{3} \int_0^{\frac{\pi}{3}}\int_0^{\frac{1}{2 \sin(\theta + \frac{\pi}{3})}} e^{-\mu \sigma^2 T h^\alpha} \cdot \prod_{\substack{(n,m) \in \\ \mathbb{Z}^2 \setminus \{(0,0)\}}}\left(\frac{1}{1 + T}\right)dt d\theta.\end{equation}

Next, given that $T > 0$, 
\begin{equation} \prod_{\substack{(n,m) \in \\ \mathbb{Z}^2 \setminus \{(0,0)\}}}\left(\frac{1}{1 + T}\right) = 0, \end{equation} 
and thus bring the expression (\ref{interm_t7_proof}) to zero.
\end{proof}

Similar to the previous scenarios, the following theorem shows that the density countering condition also applies to 2D regular network.

\begin{theorem}
\label{theorem_8}
The coverage probability of a noise-less 2D regular network for the path-loss model function $l_1(r)$ does not change when $\lambda h^2 = c$ and $c$ is constant.
\end{theorem}

\begin{proof} The proof is based on the proof of Theorem \ref{theorem_7}. By substituting $\sigma^2 = 0$ and $\lambda h^2 = c$ in (\ref{lemma_5_pc}), we obtain the following expression which is constant for a constant $c$ 
\begin{equation*} \lim_{\lambda \to \infty} p_c^{\{2,1\}} (T,\lambda, \alpha) = 2 \sqrt{3} \int_0^{\frac{\pi}{3}}\int_0^{\frac{1}{2 \sin(\theta + \frac{\pi}{3})}}  \prod_{\substack{(n,m) \in \\ \mathbb{Z}^2 \setminus \{(0,0)\}}} \left(\frac{1}{1 + T \beta(n,m)} \right) dt d\theta, \end{equation*}
where 
\begin{equation*}
\beta(n,m) =\bigg(\frac{t + \frac{\sqrt{3}}{2} c }{(m+\frac{n}{2} - \sqrt{t}\cos\theta)^2 + (n\frac{\sqrt{3}}{2} - \sqrt{t}\sin\theta)^2 + \frac{\sqrt{3}}{2} c}\bigg)^{\frac{\alpha}{2}}.
\end{equation*}

\end{proof}

In contrast to Lemma \ref{lemma_2}, the Laplace transform of the interference power $\mathcal{L}_{I_{r,\theta}}\left(s\right)$ in Lemma \ref{lemma_4} does not have a closed form expression. If desirable, the coverage probability under $l_0(r)$, $l_1(r)$ may be calculated numerically.

In Table \ref{tab:summary_1}, we summarize the conditions when SINR invariance property holds. For the considered scenarios  when the SINR invariance property does not hold, we show that the coverage probability tends to zero.

\begin{table}
\centering
\caption{Summary of conditions when SINR invariance property holds for regular/irregular 1D/2D networks}
	\begin{tabular}{|p{0.8cm}|p{1.4cm}|p{1.4cm}|p{1.4cm}|p{1.4cm}|}
		\hline
		\multirow{2}{4em}{Path-loss models} & \multicolumn{2}{|c|}{Irregular} & \multicolumn{2}{|c|}{Regular}\\
		\cline{2-5}
		& {1D} & {2D} & {1D} & {2D} \\ \hline
		{$l_0(r)$} & {holds} & {holds \cite{andrews2011tractable}} & {holds} & {holds} \\ \hline
		{$l_1(r)$} & {$\lambda h=c$} & {$\lambda h^2=c$} & {$\lambda h=c$} & {$\lambda h^2=c$} \\
		\hline
		\multicolumn{5}{|l|}{Note: $c$ is constant}\\
		\hline
	\end{tabular}

	\label{tab:summary_1}
\end{table}

\section{Average Achievable Rate Analysis}
\label{sec:aara}
In the following section, we focus on the analysis of the mean achievable data rate over a cell. More specifically, we compute the ergodic capacity which measures the long-term achievable rate averaged over all channel and network realizations \cite{li2001capacity}.

\begin{definition}
\label{definition_2} 
The average ergodic rate achievable over a cell in the downlink, assuming Rayleigh fading for desired signal, can be represented as (see Theorem 3 in \cite{andrews2011tractable}) 
\begin{multline} 
\label{tau_generic}
\tau^{\{d,i\}} (\lambda, \alpha) \overset{\scriptscriptstyle\Delta}{=} \frac{1}{\ln{2}\,}\mathbb{E}[\ln(1+SINR)]
= \frac{1}{\ln{2}\,} \int_{\mathbb{R}^d} f_r(r) \int_{\gamma_0}^\infty e^{-\frac{\mu \sigma^2 (e^t -1)}{l(r)}} \mathcal{L}_{I_r}\left(\frac{\mu (e^t -1)}{l(r)}\right) dt dr \\
= \frac{1}{\ln{2}\,} \int_{\gamma_0}^\infty p_c^{\{d,i\}}(e^t - 1, \lambda, \alpha) dt
\end{multline}

where $d$ is the dimension of the Euclidean space and $i \in \{0, 1\}$ indicates the path-loss model function. The quantity $\gamma_0$ is the minimum working SINR. In this paper, we set $\gamma_0=0$ to investigate the performance upper bound.
\end{definition}

The average achievable rate per cell can be computed numerically using Definition \ref{definition_2} with formulas for the coverage probability. However, since we have derived some closed form expressions for the coverage probability under certain conditions, further insights can be obtained for the per cell average achievable rate. In the following, we focus on deriving the average achievable rate per cell for these conditions.

\subsection{Irregular network}
\label{ar_irregular}
The average ergodic rate in the downlink of noise-less irregular network for path-loss model $l_1(r)$ for 1D and 2D can be respectively expressed as
\begin{multline} \tau^{\{1,1\}} (\lambda, 2) = \frac{2\lambda}{\ln{2}\,}\int_0^\infty \int_0^\infty e^{-2\lambda r} \cdot e^{2 \lambda \frac{(e^t -1)(h^2 + r^2) (\arctan \left( r ( (e^t-1) (h^2 + r^2) + h^2 )^{-1/2}\right) -\pi/2 )   }{ \sqrt{(e^t -1) (h^2 + r^2) + h^2}} }dr\, dt \end{multline}
and
\begin{equation} \tau^{\{2,1\}} (\lambda, \alpha) =  \frac{1}{\ln{2}\,} \int_0^\infty \frac{e^{-\pi\lambda h^2(\rho_{2}(e^t-1,\alpha))} }{1+\rho_{2}(e^t-1,\alpha)} dt. \end{equation}

The above results are obtained by substituting (\ref{eqn_pc_2d1_first}) into (\ref{tau_generic}) and (\ref{eqn_pc_2d2_first}) into (\ref{tau_generic}), as well as setting $\sigma^2 = 0$.

Note that by simply setting $h = 0$ in the equations above, we can obtain the average ergodic rate for path-loss $l_0(r)$ as
\begin{equation} \label{1d_irr_tau}\tau^{\{d,0\}} (\lambda, \alpha) =  \frac{1}{\ln{2}\,} \int_0^\infty \frac{1}{1+\rho_d(e^t-1,\alpha)} dt. \end{equation}

As we have learned from the previous section the SINR invariance property holds for the path-loss model $l_0(r)$ but not for $l_1(r)$. From the above, we can see that the average rate per cell is unaffected by the BS density for $l_0(r)$ but it is affected by the BS density for $l_1(r)$ (see also Figs. \ref{fig:aer_1d}-\ref{fig:as_conjecture}). Based on Definition \ref{definition_2} and Theorem \ref{theorem_1}, we further know that the average rate per cell approaches zero eventually. Interestingly, it is possible to counter this decay by tuning the antenna height according to BS density in order to maintain per cell average rate. This can be observed by the countering condition described in Theorem \ref{theorem_2} in conjunction with Definition \ref{definition_2}. 

\begin{figure*}[!ht]
 \begin{minipage}[l]{1.0\columnwidth}
	\centering
	\includegraphics[scale=0.50]{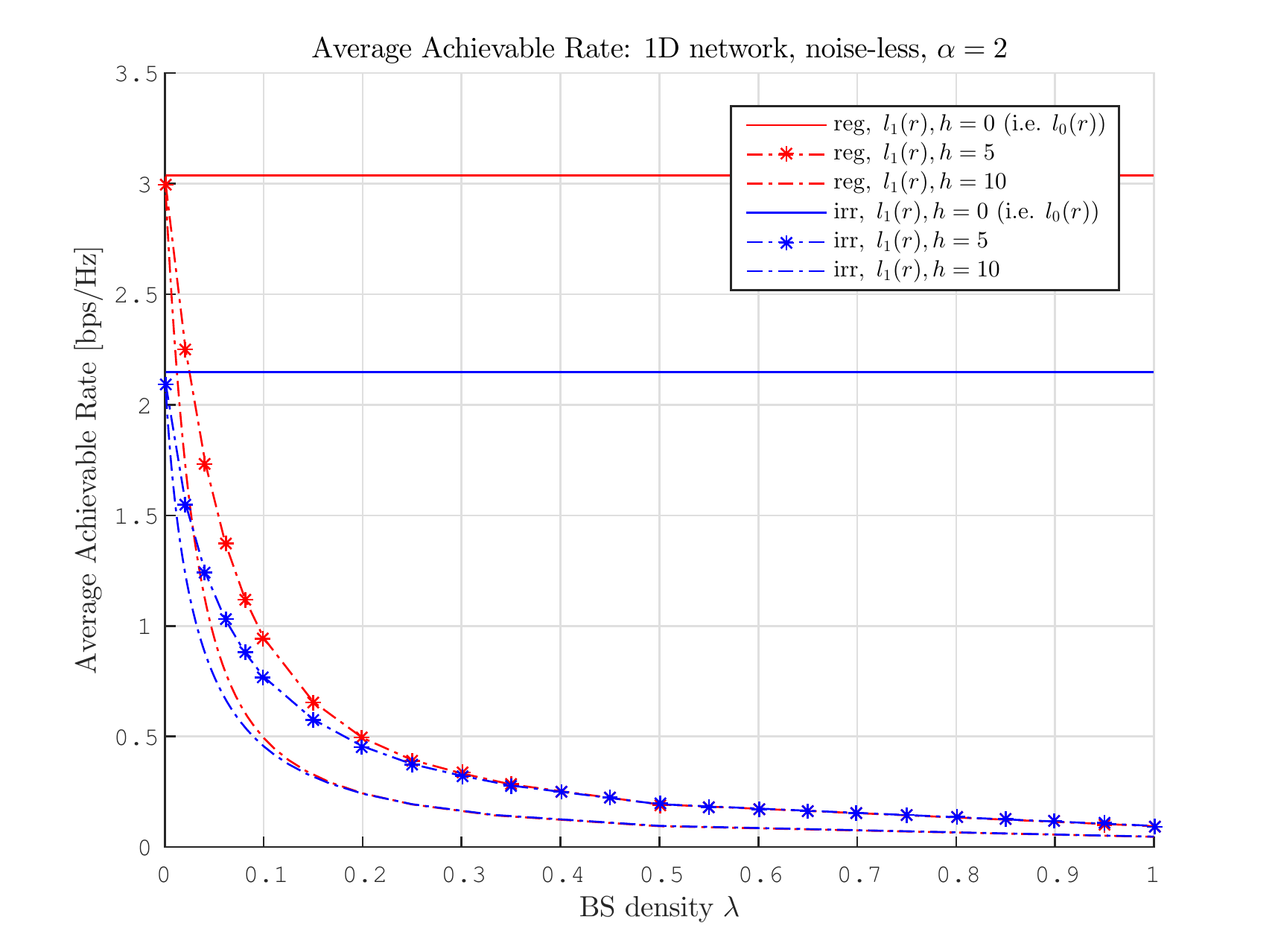}
	\caption{Impact of BS density $\lambda$ and $h$ on average ergodic rate $\tau$ of noise-less 1D network for different path-loss models with $\alpha=2$. The path-loss parameter $h$ affects the rate of decay.}
	\label{fig:aer_1d}
 \end{minipage}
 \hfill{}
 \begin{minipage}[r]{1.0\columnwidth}
	\centering
	\includegraphics[scale=0.50]{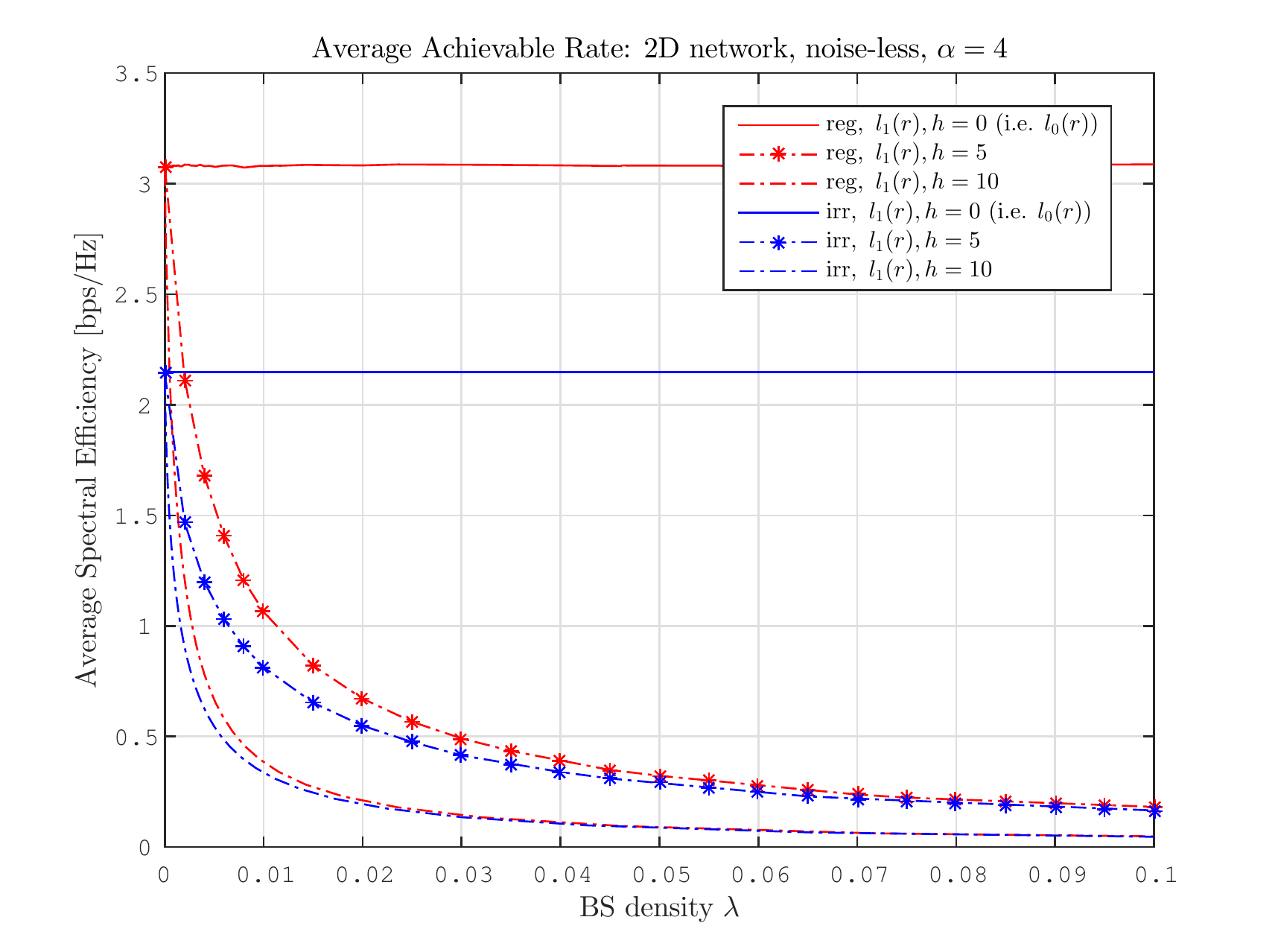}
	\caption{Impact of BS density $\lambda$ and $h$ on average ergodic rate $\tau$ of noise-less 2D network for different path-loss models withS $\alpha=4$. The path-loss parameter $h$ affects the rate of decay.}
	\label{fig:as_conjecture}
 \end{minipage}
\end{figure*}

\subsection{Regular 1D network}
In the previous section, we have derived a closed form coverage probability expression for noise-less 1D network under $l_0(r)$ path-loss model with $\alpha=2$. By substituting (\ref{theorem_3_eq}) into (\ref{tau_generic}) and employing a change of variables $x = 2\pi \lambda r$, $k=\sqrt{e^t - 1}$, we obtain the average rate per cell as
\begin{equation} \tau^{\{1,0\}} (\lambda, 2) =  \frac{2}{\ln{2}\,\pi} \int_0^\infty \int_0^\pi \frac{k(\cos(x)-1)}{\cos(x)-\cosh(xk)} dx dk.  \end{equation}

The above expression indicates a constant per cell average achievable rate regardless of BS density which is similar to what was shown in \cite{andrews2011tractable} for irregular networks. This expression also allows us to numerically compute the per cell average rate. We have found that regular network does indeed give a higher per cell average rates than that of the irregular network. For regular network, the approximate average rate per cell is $\approx 3.037 \mathrm{bits/sec/Hz}$. For irregular network, the approximate average rate per cell can be computed by (\ref{1d_irr_tau}) which gives $\approx 2.148 \mathrm{bits/sec/Hz}$. By comparing the above results of regular and irregular networks, performing site planning results in an approximate gain of $1.414$ times higher achievable rate. 

Next, for path-loss model $l_1(r)$, by substituting (\ref{theorem_3_eq}) into (\ref{tau_generic}) and employing a change of variables $k=e^t -1$, the average achievable rate per cell with $\alpha = 2$ can be expressed as
\begin{equation} \tau^{\{1,1\}} (\lambda, 2) = \frac{1}{\ln{2}\,\pi} \cdot \int_0^\infty \int_0^\pi \frac{\cos(x)-\cosh(2\pi\lambda h)}{\cos(x)-\cosh(\sqrt{((2\pi\lambda h)^2(1+ k) + x^2 k})} dx dk. \end{equation}

From the last result, we can see that the average rate per cell for $l_1(r)$ with $\alpha = 2$ depends on BS density as also seen in Figure \ref{fig:aer_1d}. However, by setting $h=\frac{c}{\lambda}$ where $c$ is some constant positive value, we can maintain a constant average rate per cell regardless of BS density. With Theorem \ref{theorem_5}, the above argument can be extended to a general condition of $\alpha>1$.

From the last result, by considering $l_1(r)$ with $\alpha = 2$, we see that the average achievable rate per cell depends on BS density (see also Fig. \ref{fig:aer_1d}). Similar to the irregular network, the average achievable rate per cell approaches zero when BS density goes to infinity. However, by setting $h=\frac{c}{\lambda}$ where $c$ is some constant positive value, we can maintain a constant average achievable rate per cell regardless of BS density. This argument is also valid for general condition of $\alpha>1$ by observing Theorems \ref{theorem_4} and \ref{theorem_5} in conjunction with Definition \ref{definition_2}.

\subsection{Regular 2D network}

As the probability of coverage for 2D regular network does not yield any closed form expression, we seek numerical computation to calculate its average achievable rate per cell. We present the numerical results in Figure \ref{fig:as_conjecture} illustrating the impact of $\lambda$ and $h$ on the average ergodic rate per cell $\tau^{\{2, 1\}}(\lambda, \alpha)$. As can be seen, the rate of decay of $\tau^{\{2, 1\}}(\lambda, \alpha)$ for an increasing $\lambda$ depends on $h$ and $\tau^{\{2, 1\}}(\lambda, \alpha)$ tends to zero. This is in line with the results for other network configurations discussed in the previous subsections.
 
Additionally, based on Theorem \ref{theorem_8} and Definition \ref{definition_2}, it can be easily shown that the average achievable rate per cell in the noise-less 2D regular network under path-loss model $l_1(r)$ does not change when $\lambda h^2 = c$ and $c$ is constant. In other words, it is possible to maintain per cell average achievable rate by countering the increase in BS density through lowering the antenna height accordingly.

Furthermore, Theorem \ref{theorem_6} and Definition \ref{definition_2} indicate that the 2D noise-less regular network under $l_0(r)$ exhibits SINR invariance property for all conditions. Consequently, the average rate per cell for $l_0(r)$ does not depend on BS density.

\section{Performance Limits in UDN}
In the previous section, we have derived the coverage probability and average achievable rate per cell for various scenarios. In this section, we investigate how extreme densification affects gains from site planning and its impact on ASE.

\begin{figure*}[!ht]
 \begin{minipage}[l]{1.0\columnwidth}
     \centering
     \includegraphics[scale=0.50]{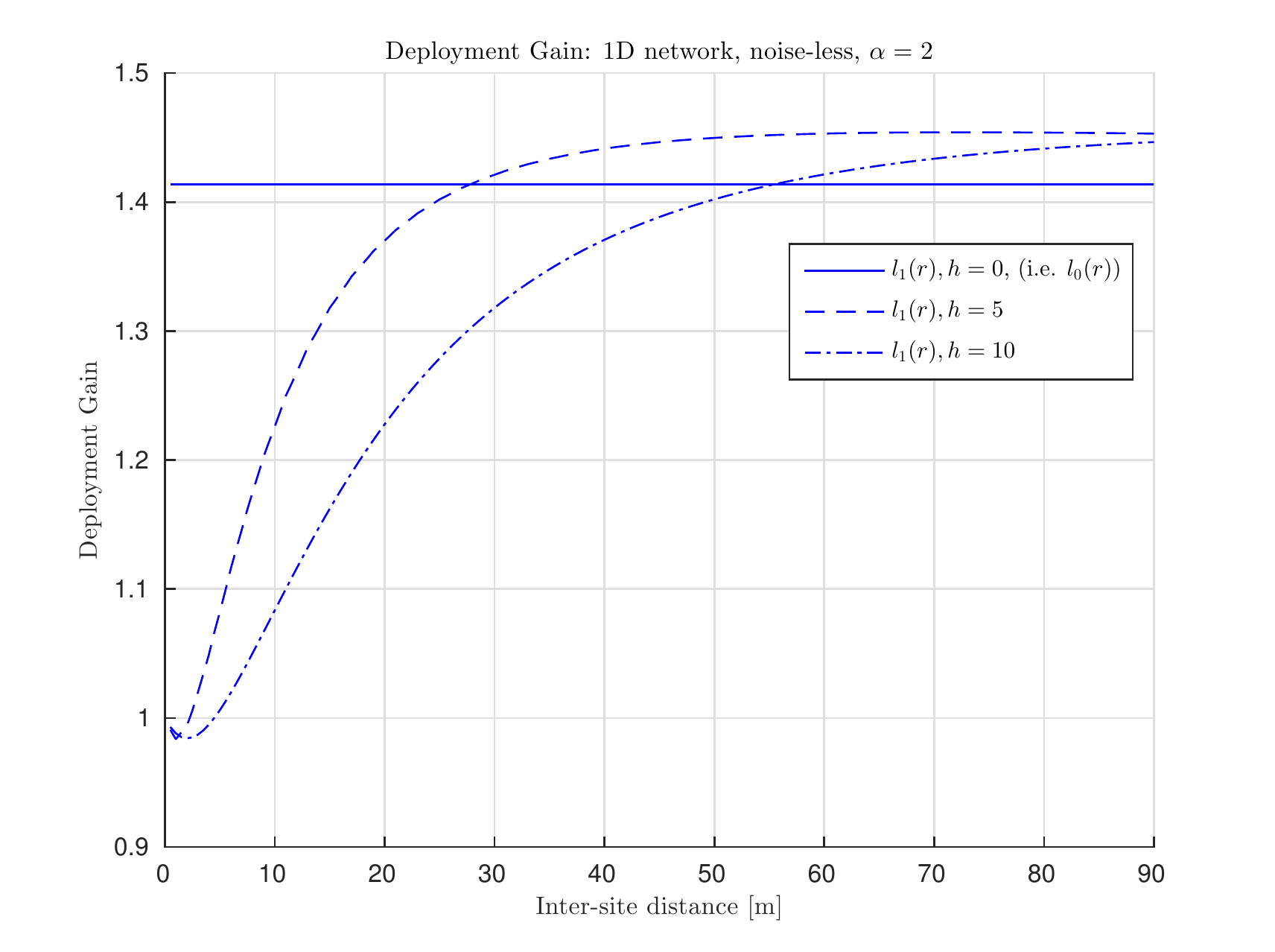}
     \caption{Deployment gain for 1D noise-less network with $\alpha=2$.}\label{fig:densification}
 \end{minipage}
 \hfill{}
 \begin{minipage}[r]{1.0\columnwidth}
     \centering
     \includegraphics[scale=0.50]{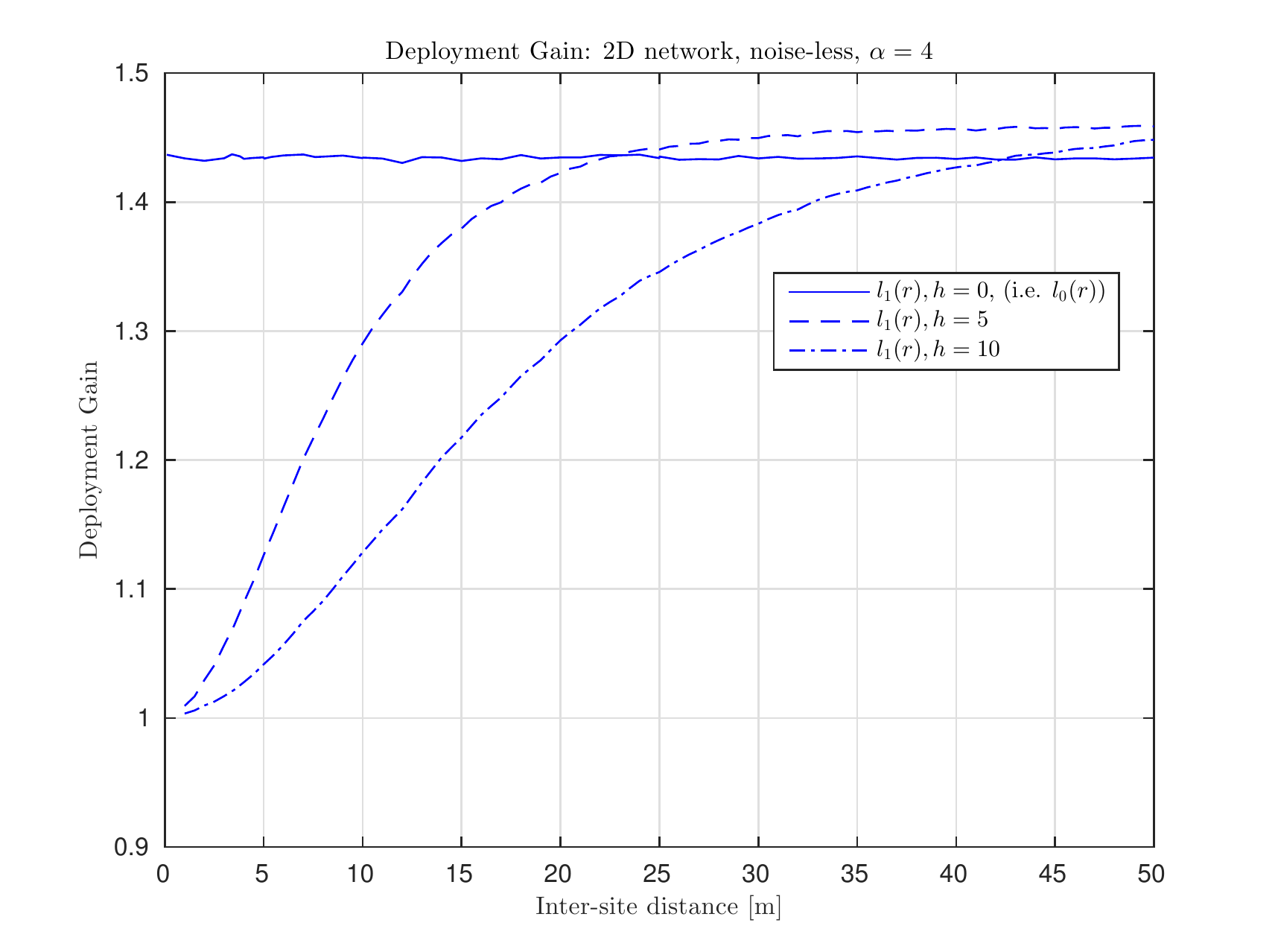}
     \caption{Deployment gain for 2D noise-less network with $\alpha=4$.}\label{fig:densification2}
 \end{minipage}
\end{figure*}

\subsection{Deployment gain}

We define deployment gain as the ratio of average achievable rate per cell for regular network to that of the irregular network. In Section \ref{sec:aara}, we have presented the average achievable rate per cell for different scenarios in Figures \ref{fig:aer_1d} and \ref{fig:as_conjecture}. When comparing the rate performance between the regular and irregular networks, we made the following observations: (i) for the case where SINR invariance property holds, the deployment gain remains constant regardless of BS density, (ii) for the case where SINR invariance property does not hold, the deployment gain depends on BS density and no gain is observed when BS density tends to infinity, (iii) in general, we achieve a performance gain when the deployment of BSs follows a regular pattern. To illustrate our observations, we further plot the deployment gain in Figure \ref{fig:densification} for 1D network with $\alpha=2$ and Figure \ref{fig:densification2} for 2D network with $\alpha=4$. To focus on the performance in ultra dense region, our plots use average inter-site distance on x-axis.

From Figure \ref{fig:densification}, we first see a constant deployment gain of approximately 1.414 for 1D network when SINR invariance property holds. When SINR invariance property does not hold, the deployment gain varies based on BS density. We see from the figure that the deployment gain increases as the average inter-site distance decreases. After peaking at a certain point, the gain appears to decrease to zero when the inter-site distance approaches zero. In other words, there appears to be no gain when BS density is extremely high. Interestingly, there is a small region of inter-site distance when deployment gain value falls below one, indicating negative impact of site planning. In Figure \ref{fig:densification2}, we observe similar behavior for the 2D network. In the following, we show that deployment gain vanishes when BS density approaches infinity.

\begin{proposition}
\label{proposition_10a}
System performance under $l_1(r)$ does not depend on BS deployment strategy (i. e. no gain from site planning) as $\lambda \to \infty$.
\end{proposition}

\begin{proof}
We first study the distance between a user at a particular point and an arbitrary interfering BS. In regular network, since the location of the interfering BS is deterministic, the distance between the two is also deterministic. Let this distance be $R$. In irregular network, the location of the BS follows the PPP. The distance is a random variable, say $\tilde{R}$, which has the following statistical properties \cite{haenggi2009interference} for 1D network
\begin{equation} \mathbb{E}[\tilde{R}] = \frac{n}{2\lambda}\end{equation}
\begin{equation} \mathrm{var}(\tilde{R}) = \frac{n}{4\lambda^2},\end{equation}
and 2D network
\begin{equation} \mathbb{E}[\tilde{R}] = \frac{1}{\sqrt{\lambda \pi}} \frac{\Gamma(n + \frac{1}{2})}{\Gamma(n)}\end{equation}
\begin{equation} \mathrm{var}(\tilde{R}) = \frac{1}{\lambda \pi} \left(n - \left(\frac{\Gamma(n + \frac{1}{2})}{\Gamma(n)}\right)^2\right).\end{equation}

Considering an interfering BS located at a particular distance away from the user. For regular network, this distance is a fixed value $R$. For irregular network, this distance is $\tilde{R}$ with mean value of $R$. Notice that when $\lambda \to \infty$, the variance of $\tilde{R}$ tends to zero. In other words, when $\lambda$ continues to increase, the variance of this distance continues to decrease which reduces the deployment difference between irregular and regular networks.
For a regular network, the interference caused by this BS is $g \cdot l(R)$ where $g$ is a random variable describing the channel fading and $l(\cdot)$ is the path-loss expression. For an irregular network, the interference caused by this BS is $g \cdot l(\tilde{R})$. Given that $l_1(r)$ path-loss model is bounded, we have $l_1(r)<\infty$ for an arbitrary $r \geq 0$. When $\lambda \to \infty$, the variance of $\tilde{R}$ decreases to zero, we see that the interference characteristic caused by the BS in the irregular network converges to that of the regular network, that is $g l(\tilde{R}) \to g l(R)$. This shows that system performance under $l_1(r)$ for irregular network converges to that of the regular network when $\lambda \to \infty$.
\end{proof}
The above proposition can be extended to any bounded path-loss model. For the unbounded path-loss model, we show in Section \ref{sec:aara} that a constant performance gain is achieved when the deployment of BSs follows a regular pattern.

\subsection{Average Area Spectral Efficiency}

The ASE is a measure of the overall rate over a network area and is defined by

\begin{equation} ASE = \lambda \cdot \tau^{\{d, i\}}(\lambda, \alpha). \end{equation}

ASE may increase or decrease as network densifies. This depends on whether the decay of rate per cell as BS density increases, can be countered by spatial reuse, and therefore it is of interest to investigate whether continued network densification will still lead to ASE improvement. The finding may help network operators to optimize their investments in the infrastructure and identify when further densification may not be beneficial. 

In the previous section, we showed that when SINR invariance does not hold, the average achievable rate per cell goes to zero as $\lambda \to \infty$. However, as we show in the following, the average ASE does not necessarily go to zero. To support this, we formulate the following Theorem.

\begin{figure*}[!ht]
	\begin{minipage}[l]{1.0\columnwidth}
		\centering
		\includegraphics[scale=0.52]{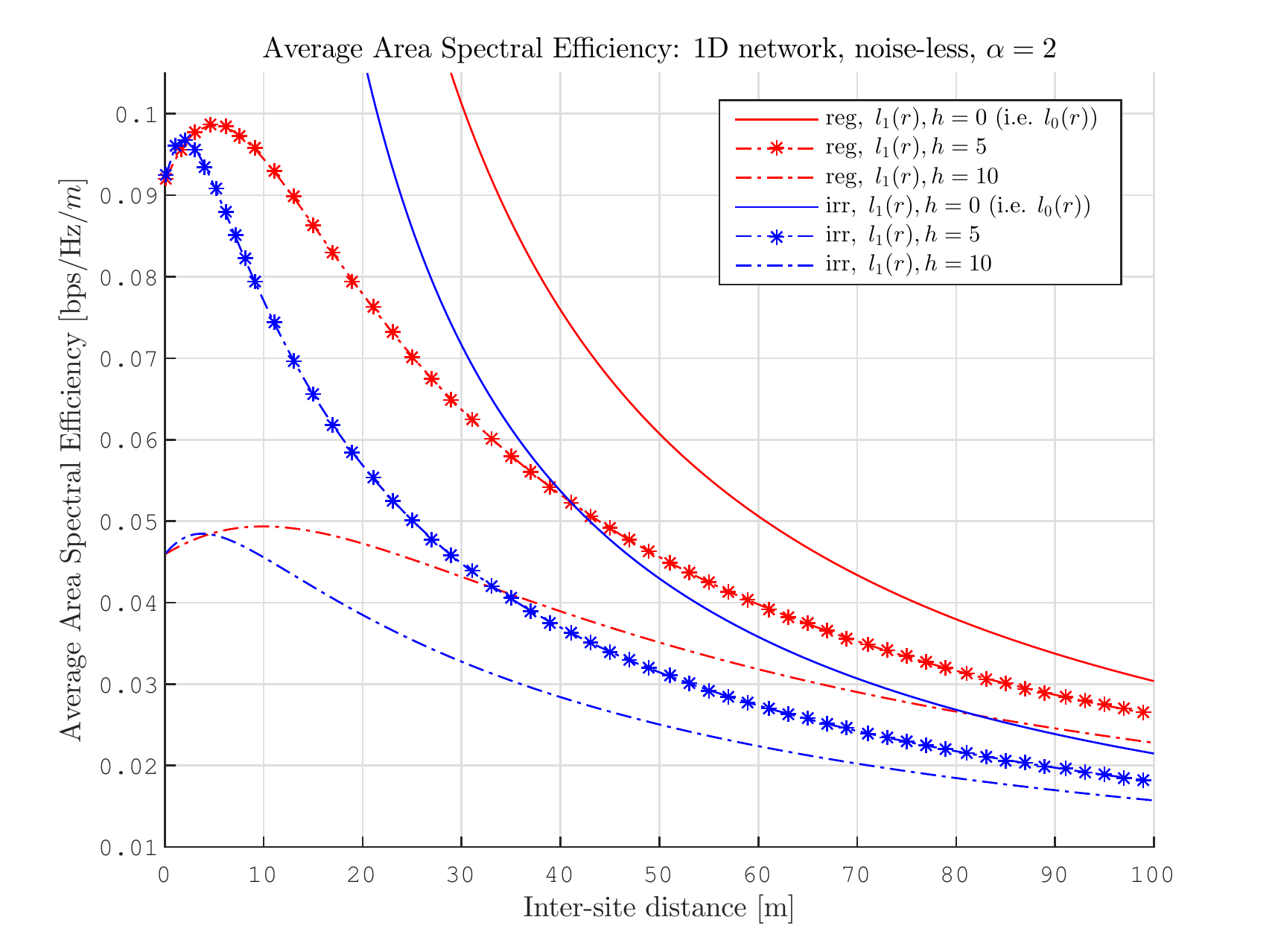}
		\caption{Average Area Spectral Efficiency for noise-less 1D networks for $\alpha=2$}
		\label{fig:ase_1d}
	\end{minipage}
	\hfill{}
	\begin{minipage}[r]{1.0\columnwidth}
		\centering
		\includegraphics[scale=0.52]{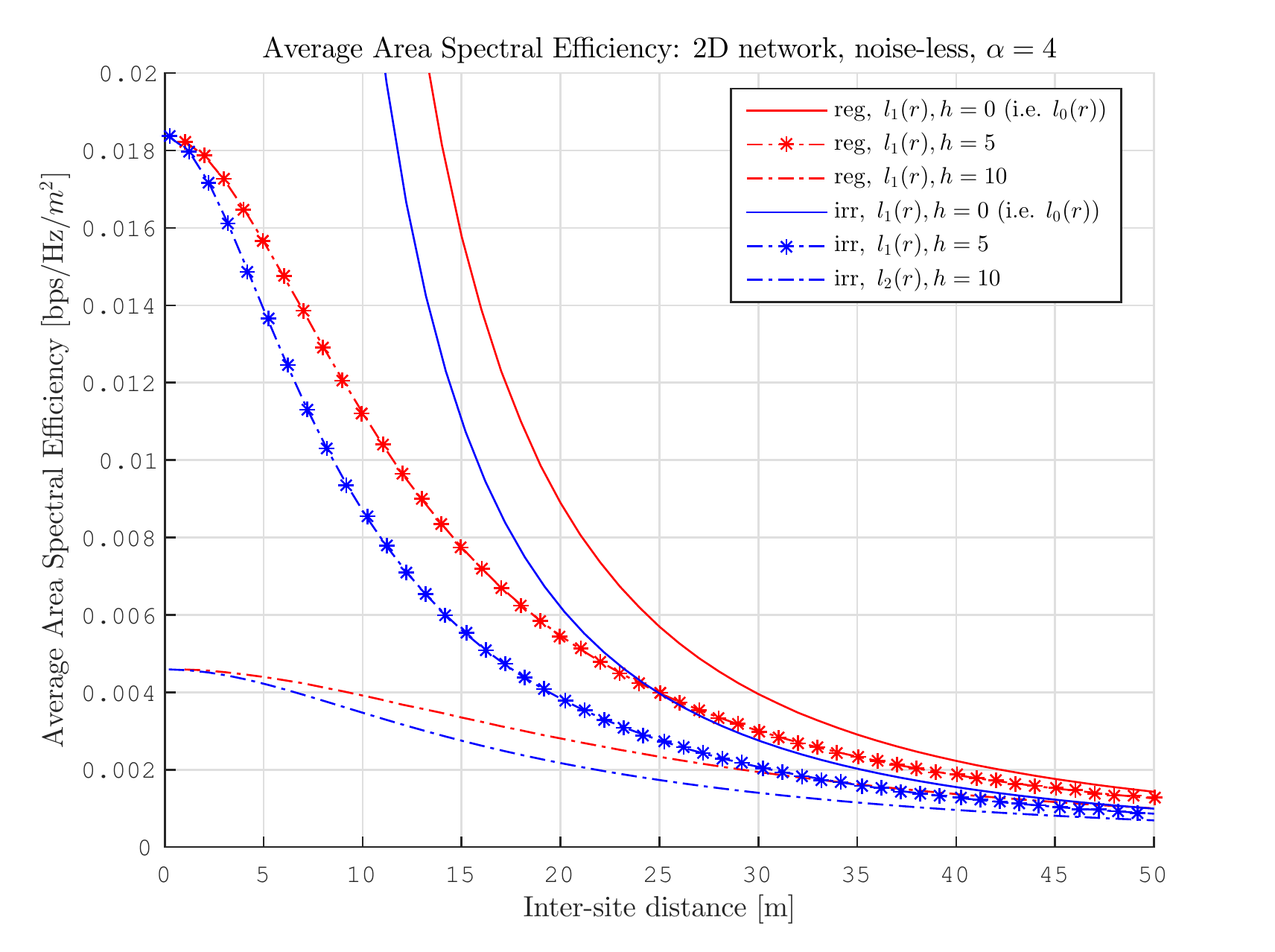}
		\caption{Average Area Spectral Efficiency for noise-less 2D networks for $\alpha=4$}
		\label{fig:ase_2d}
	\end{minipage}
\end{figure*}

\begin{theorem}
\label{theorem_aase}
Average ASE of cellular networks under $l_1(r)$ for $h>0$ as $\lambda \to \infty$, given $\alpha > 1$ for 1D or $\alpha > 2$ for 2D converges to a non-zero value. 
\end{theorem}

\begin{proof}
Proof is given in Appendix \ref{theorem_aase_proof}. 
\end{proof}
As can be seen in the proof of Theorem \ref{theorem_aase}, the ASE of UDN under $l_1(r)$ has the following lower bounds when $\lambda \to \infty$, for 1D network
\begin{equation} \frac {1}{2\ln{2}\,\,h\left(1 + \frac{\Gamma(\alpha-1)}{\Gamma(\alpha)}\right)}, \end{equation}
and 2D network
\begin{equation} \frac {2\sqrt{3}}{12\ln{2}\,\,h^2\left(1+2\,\frac{\Gamma(\alpha-2)}{\Gamma(\alpha-1)}\right)}. \end{equation}

These lower bounds depend on $h$ and $\alpha$. They show that by lowering $h$, the ASE can be improved. This is very encouraging as network operators could theoretically provide services even with over-densified networks. In the extreme case when $h \to 0$, we have ASE approaching infinity. This is true since when $h=0$, the path loss model $l_1(r)$ reduces to $l_0(r)$, and SINR invariance property holds for this scenario. 

It is also worth noting here that the lower bounds also depend on $\gamma_0$ which is the minimum working SINR (see Definition \ref{definition_2}). The above results consider $\gamma_0=0$ setting. A higher $\gamma_0$ will lower the per cell rate and thus will reduce the ASE. With some sufficiently high value of $\gamma_0$, the ASE reduces to zero as indicated in \cite{ding2017performance}.

In the following, we show ASE performances for some specific scenarios using the above developed results. We derive exact ASE performance for 1D network when $\lambda \to \infty$ for various $\alpha$ values. 

\begin{proposition}
\label{proposition_12}
Average area spectral efficiency (in terms of $\mathrm{bits/Hz/m^2}$) of 1D network for $\alpha=\{2, 4, 6\}$ under $l_1(r)$ path-loss model when $\lambda \to \infty$ is

\begin{equation}\lim_{\lambda \to \infty} \lambda\,\tau^{\{1,1\}}(\lambda, 2)  = \frac{1}{\ln{2}\,\pi h},\end{equation}

\begin{equation}\lim_{\lambda \to \infty} \lambda\,\tau^{\{1,1\}}(\lambda, 4)  = \frac{2}{\ln{2}\,\pi h},\end{equation}

\begin{equation}\lim_{\lambda \to \infty} \lambda\,\tau^{\{1,1\}}(\lambda, 6)  = \frac{8}{3 \ln{2}\,\pi h}.\end{equation}

\end{proposition}

\begin{proof} Proof is given in Appendix \ref{proposition_12_proof}.\end{proof}

For 2D network, as a closed form expression for Laplace transform of $I_r(s)$ is unavailable, we seek to obtain a tighter bound for specific $\alpha$ values as follows.
\begin{proposition}
\label{proposition_2d_2}
The upper and lower bounds on the average area spectral efficiency (in terms of $\mathrm{bits/Hz/m^2}$) of 2D network for $\alpha=4$ under $l_1(r)$ path-loss model, when $\lambda \to \infty$ is

\begin{equation} \frac{2\sqrt{3}}{9\ln{2}\,h^2} \geq \lim_{\lambda \to \infty} \lambda\,\tau^{\{2,1\}}(\lambda, 4)  \geq \frac{2\sqrt{3}}{12\ln{2}\,h^2}. \end{equation} 

\end{proposition}

\begin{proof} Proof is given in Appendix \ref{proposition_2d_2_proof}.\end{proof}

\section{Conclusion}

This paper investigated the impact of the relative antenna height between BS and UE antennas on the performance of UDNs. We showed that for any nonzero relative antenna height the highly desirable SINR invariance property does not hold, indicating that the UDN performance is dependant on BS density. Moreover, the per cell average achievable rate performance of UDN decays to zero, as BS density increases. In order to manage this decay, we showed that decreasing the relative antenna heights across the network can counter the decay of per cell average achievable rate. We explicitly derived the relationship between BS density and relative antenna height showing how this can be achieved. Interestingly, appropriate adjustment of BS antenna heights enables the network to retain SINR invariance property. 

We further studied the ASE performance of the network. Despite the pessimistic conclusion related to the per cell performance shown in the literature, in this paper we showed that ASE does not necessarily decay to zero as BS density approaches infinity. A non-negligible ASE can be achieved which is dependent of the path-loss exponent and relative antenna height. This result is optimistic as it indicates that theoretically even in a highly over-densified network, the network can continue to serve with a reduced performance. 

We finally investigated the deployment gain which studies the performance improvement due to careful site planning in UDNs. We showed that, in general, we can always achieve a performance gain when the deployment of BSs follows a regular pattern. This gain remains constant as BS density increases when SINR invariance property holds and deteriorates to zero otherwise. The insights provided in this work may help network operators to optimize their investments in the infrastructure and identify when further densification or careful site selection may no longer be beneficial due to low cost-effectiveness.


%
\appendices

\section{Proof of Theorem \ref{theorem_aase}}
\label{theorem_aase_proof}

We start by stating the following Lemmas which are then used in the proof of the theorem.

\begin{lemma}
\label{lemma_appendix}
Coverage probability performance of regular and irregular networks under $l_1(r)$ with Rayleigh fading converge to that of networks without fading as $\lambda \to \infty$.
\end{lemma}

\begin{proof}
Consider an irregular network of density $\lambda'$ with Rayleigh fading and another irregular network of density $\lambda$ without fading. Using findings in \cite{zhang2014performance, madhusudhanan2012downlink}, Baccelli et al. show in \cite{baccelli2015correlated} that distributions of SINR for both networks are the same if $\lambda' = \lambda/\Gamma(\delta + 1)$ where $\delta = \frac{d}{\alpha}$. Since the distributions of SINR are the same, both networks give the same coverage probability performance. Given that the path-loss model $l_1(r)$ is bounded, with $\lambda \to \infty$, the coverage probability of the network with Rayleigh fading simply converges to that of the network without fading as $\lambda’ \to \infty$.

Using Proposition \ref{proposition_10a}, the above conclusion can be extended to the case of regular networks, thus concluding the proof.
\end{proof}

\begin{lemma}
\label{lemma_appendix2}
Given a path-loss model $l_1(r)$ and $h>0$, the ASE performance of a network when $\lambda \to \infty$ converges to the ASE performance of a network where each user is collocated with its serving BS.
\end{lemma}

\begin{proof} Let $\hat{\tau}^{\{d,i\}}(\lambda, \alpha)$ be the average achievable rate per cell of a network without fading (i.e. $g_x \equiv 1$). Let $\hat{\tau}^{\{d,i\}}_0(\lambda, \alpha)$ be the achievable rate per cell of a network without fading given both a typical user and its serving BS are located in the origin.

For 1D regular network, $\hat{\tau}^{\{1,i\}}(\lambda, \alpha)$ can be expressed as
\begin{equation}
\hat{\tau}^{\{1,i\}}(\lambda, \alpha) = \int_0^{\frac{\Upsilon}{2}}\frac{2}{\Upsilon} \cdot \log_2 \bigg(1+\frac{l(x)}{\sum_{i=1}^{\infty}l(i\,\Upsilon + x) + \sum_{i=1}^{\infty}l(i\,\Upsilon - x) + \mu \sigma^2}\bigg)dx\end{equation}
where $\Upsilon = \frac{1}{\lambda}$. By substituting $x=k\frac{\Upsilon}{2}$ to the above result, we obtain
\begin{equation}
\label{lema_eq_app2}
\hat{\tau}^{\{1,i\}}(\lambda, \alpha) =  \int_0^1\log_2 \bigg(1 +\frac{l(\Upsilon\frac{k}{2})}{\sum_{i=1}^{\infty}l(\Upsilon i + \Upsilon \frac{k}{2}) + \sum_{i=1}^{\infty}l(\Upsilon i - \Upsilon \frac{k}{2}) + \mu \sigma^2}\bigg)dk.\end{equation}

Next, by assuming a bounded path-loss model function (i.e. $l(0) < \infty)$) and applying the Monotone Convergence Theorem, we derive the following limit expression
\begin{multline}
\label{eq_lemma_appendix2_1}
\lim_{\lambda\to \infty} \hat{\tau}^{\{1,i\}}(\lambda, \alpha) = \lim_{\lambda\to \infty} \int_0^1\log_2 \bigg(1+\frac{l(0)}{\sum_{i=1}^{\infty}l(\frac{i}{\lambda}) + \sum_{i=1}^{\infty}l(\frac{i}{\lambda}) + \mu \sigma^2}\bigg)dk \\=
\lim_{\lambda\to \infty} \log_2 \bigg(1+\frac{l(0)}{\sum_{i=1}^{\infty}l(\frac{i}{\lambda}) + \sum_{i=1}^{\infty}l(\frac{i}{\lambda}) + \mu \sigma^2}\bigg)
.\end{multline}

Now we express $\hat{\tau}^{\{d,i\}}_0(\lambda, \alpha)$ for 1D regular network under a bounded path-loss model function as
\begin{equation}
\label{eq_lemma_appendix2_2}
\hat{\tau}^{\{1,i\}}_0(\lambda, \alpha) = \log_2 \bigg(1+\frac{l(0)}{\sum_{i=1}^{\infty}l(i\,\Upsilon) + \sum_{i=1}^{\infty}l(i\,\Upsilon) + \mu \sigma^2}\bigg).\end{equation}

Using (\ref{eq_lemma_appendix2_1}) and (\ref{eq_lemma_appendix2_2}), we can easily obtain the following
\begin{equation} 
\lim_{\lambda\to \infty} \lambda \hat{\tau}^{\{d,i\}}(\lambda, \alpha) = \lim_{\lambda\to \infty} \lambda \hat{\tau}^{\{d,i\}}_0(\lambda, \alpha).
\end{equation}

We can repeat the above approach for 2D networks, thus concluding the proof.
\end{proof}

In the following we provide the proof of Theorem \ref{theorem_aase}.

\begin{proof}
We begin the proof by focusing on the 1D network. Using Lemma \ref{lemma_appendix2}, we write the limit of ASE for 1D regular network without fading when $\lambda \to \infty$ as
\begin{multline}
\label{generic_1d_limit}
\lim_{\lambda\to \infty} \lambda \hat{\tau}^{\{1,i\}}_0(\lambda, \alpha)  = \lim_{\lambda\to \infty} \lambda \log_2 \bigg(1+\frac{l(0)}{\sum_{i=1}^{\infty}l(i\,\Upsilon) + \sum_{i=1}^{\infty}l(i\,\Upsilon) + \mu \sigma^2}\bigg).
\end{multline}

To derive the lower bound of the above expression for $l_1(r)$, we introduce a new path-loss model function $l_2(r) = (max(h, r))^{-\alpha}$ where $l_2(r) \geq l_1(r)$, for all $r\geq 0, h \geq 0$. By substituting $l(r) = l_1(r)$ and $l(r) = l_2(r)$ in (\ref{generic_1d_limit}) it can be easily seen that
\begin{equation}\lim_{\lambda\to \infty} \lambda \hat{\tau}^{\{1,1\}}_0(\lambda, \alpha) \geq \lim_{\lambda\to \infty} \lambda \hat{\tau}^{\{1,2\}}_0(\lambda, \alpha). \end{equation}

The above inequality can be further simplified into the following expression 
\begin{multline}
\lim_{\lambda\to \infty} \lambda \hat{\tau}^{\{1,1\}}_0(\lambda, \alpha)  \geq \lim_{\lambda\to \infty} \lambda \log_2 \bigg(1+ \frac{h^{-\alpha}}{\mu \sigma^2 + 2 \big(n h^{-\alpha} + \Upsilon^{-\alpha}\zeta(\alpha, 1+n)\big)}\bigg),
\end{multline}
where
\begin{equation} \zeta(\alpha, 1+n) = \sum_{i=n+1}^{\infty}\frac{1}{i^{\alpha}} \end{equation}
is the Hurwitz Zeta function and $n = \lfloor \frac{h}{\Upsilon} \rfloor$.

We substitute $t = \frac{h}{\Upsilon}$ and without loss of generality, as $\lambda \to \infty$ (or $\Upsilon \to 0$), we can safely assume $t \in \mathbb{N}$ which leads to the following expression
\begin{equation}
\label{generic_1d_limit_2} 
\lim_{t\to \infty} \frac{t}{h}\log_2\left(1+ \frac{1} {\mu\sigma^2h^{\alpha}  + 2 \left(t\, + t^{\alpha} \zeta\left(\alpha, 1+ t \right)\right)}\right).
\end{equation}

By applying Watson's Lemma to the integral representation of the Hurwitz Zeta function we obtain the asymptotic expansion of $\zeta(\alpha, 1+ t)$ for $t \to \infty$ as

\begin{multline} \zeta(\alpha, 1 + t) = \frac{1}{\Gamma(\alpha)}  \int_{0}^{\infty} \frac{x^{\alpha - 1} e^{-(1+t)x}}{1-e^{-x}}dx  \sim \frac{\Gamma(\alpha - 1)}{\Gamma(\alpha)}(1+t)^{-\alpha + 1} + \frac{1}{2} (1+t)^{-\alpha} + O(t^{-\alpha - 1}).\end{multline}
By substituting $\zeta(\alpha, 1 + t)$ with it's asymptotic expansion and by applying the L'Hopital rule we can rewrite (\ref{generic_1d_limit_2}) and simplify it to obtain the following result

\begin{align} 
&\lim_{t \to \infty} \frac{t}{h} \log_2\Bigg[1+ \bigg(\mu\sigma^2\, h^{\alpha} + 2 \Big( t+ t^{\alpha}\frac{\Gamma(\alpha-1)}{\Gamma(\alpha)} \cdot (1+t)^{-\alpha+1} + \frac{1}{2}(1+t)^{-\alpha}\Big)\bigg)^{-1}\Bigg] & \nonumber \\ 
& = \lim_{k \to 0^+} \frac{1}{k h} \log_2\Bigg[1+  k\cdot\bigg(k \mu \sigma^2 h^{\alpha} + 2 \Big( 1+ \frac{\Gamma(\alpha-1)}{\Gamma(\alpha)}\cdot (1+k)^{-\alpha+1} + \frac{k}{2}(1+k)^{-\alpha}\Big)\bigg)^{-1}\Bigg] & \nonumber \\ 
& = \frac {\Gamma(\alpha)}{2 \ln2\,h\left(\Gamma(\alpha) + \Gamma(\alpha-1)\right)}.
\end{align}

From the above result it can be easily seen that the lower bound of ASE as $\lambda \to \infty$ for 1D regular network without fading under $l_1(r)$ is greater than zero and the lower bound depends on $h$ and $\alpha$. Next, using Proposition \ref{proposition_10a} and Lemma \ref{lemma_appendix} it can be shown that this result can be extended to 1D irregular network and 1D regular network with Rayleigh fading, thus concluding the proof for 1D network.

For 2D network, similar to 1D network, we use Lemma \ref{lemma_appendix2} and write the limit of ASE for 2D regular network without fading when $\lambda \to \infty$ as

\begin{equation} 
\label{generic_2d_limit}
\lim_{\lambda\to \infty} \lambda \hat{\tau}^{\{2, i\}}_0(\lambda, \alpha) 
 =\lim_{\lambda\to \infty} \lambda \log_2\left(1+ \frac{l(\|b_0\|)}{\mu\sigma^2 + \sum_{\Phi^{HEX} \setminus b_o} l(\|b_i\|)}\right).
\end{equation}

Comparing the limit of ASE under the path-loss models between $l_1(r)$ and $l_2(r)$, we get
\begin{equation} \lim_{\lambda\to \infty} \lambda \hat{\tau}^{\{2,1\}}_0(\lambda, \alpha) \geq \lim_{\lambda\to \infty} \lambda \hat{\tau}^{\{2,2\}}_0(\lambda, \alpha). \end{equation}

It can be also easily seen that $I_0 = \sum_{\Phi^{HEX} \setminus b_o} l(\|b_i\|)$ in (\ref{generic_2d_limit}) is the cumulated interference from BSs deployed according to $\Phi^{HEX}$ given that a typical user and its serving BS are located at the origin. As this expression does not have a closed-form we use the expressions for lower and upper bounds as presented in \cite{haenggi2009interference}
\begin{equation} \sum_{k=1}^{\infty} 6k \cdot l(\Upsilon k)< I_0 < 6\,l(\Upsilon) + \sum_{k=2}^{\infty} 6k \cdot l(\Upsilon \sqrt{3}/2\, k)\end{equation}
By substituting $l(r) = l_2(r)$ and using the definition of Hurwitz zeta function, we can simplify the expressions for lower and upper bound as presented below

\begin{equation*} \sum_{k=1}^{n_1} 6k \cdot h^{-\alpha} + 6\Upsilon^{-\alpha} \zeta(\alpha -1, 1+n_1) < I_0\end{equation*}
\begin{equation*} \sum_{k=1}^{n_2} 6k \cdot h^{-\alpha} + 6\Upsilon^{-\alpha}(\sqrt{3}/2)^{-\alpha}\zeta(\alpha-1,1+n_2) > I_0 \end{equation*}
where $n_1 = \lfloor \frac{h}{\Upsilon} \rfloor$, $n_2 = \lfloor \frac{2h}{\sqrt{3}\Upsilon} \rfloor$ and $n_1 \geq 0$, $n_2\geq 1$.

By taking the upper bound of the cumulated interference we derive the lower bound of $\lim_{\lambda \to \infty} \lambda \hat{\tau}^{\{2, 2\}}_0(\lambda, \alpha)$ as presented below 

\begin{multline} 
\label{proposition_2d_proof_eq}
\lim_{\lambda\to \infty} \lambda\,\hat{\tau}^{\{2,2\}}_0(\lambda, \alpha) > \\
\lim_{\Upsilon\to 0^+} \frac{2}{\Upsilon^2\sqrt{3}}\log_2\left(1+ \frac{h^{-\alpha}} {\mu\sigma^2  + 3(n_2+1)n_2h^{-\alpha} + 6 \Upsilon^{-\alpha} (\sqrt{3}/2)^{-\alpha} \zeta\left(\alpha-1, 1+ n_2\right)}\right) = \\ \lim_{\Upsilon\to 0^+} \frac{2}{\Upsilon^2\sqrt{3}}\log_2\left(1+ \frac{1} {\mu\sigma^2 h^{\alpha}  + 3(\lfloor \frac{2h}{\sqrt{3}\Upsilon} \rfloor+1)\lfloor \frac{2h}{\sqrt{3}\Upsilon} \rfloor + 6 \left(\frac{h}{\Upsilon}\right)^{\alpha} (\sqrt{3}/2)^{-\alpha} \zeta\left(\alpha-1, 1+ \lfloor \frac{2h}{\sqrt{3}\Upsilon} \rfloor\right)}\right).
\end{multline}

Note that it is also the lower bound of $\lim_{\lambda \to \infty} \lambda \hat{\tau}^{\{2, 1\}}_0(\lambda, \alpha)$. We substitute $k = \frac{2 h}{\sqrt{3}\Upsilon}$ in (\ref{proposition_2d_proof_eq}). As $\lambda \to \infty$, we have $\Upsilon \to 0$ and hence we can safely assume that $k \in \mathbb{N}$ which leads to the following expression

\begin{multline}
\label{proposition_2d_proof_eq_2}
\lim_{\lambda\to \infty} \lambda\,\hat{\tau}^{\{2, 1\}}_0(\lambda, \alpha) \geq \lim_{\lambda\to \infty} \lambda\,\hat{\tau}^{\{2, 2\}}_0(\lambda, \alpha) >  \\ \lim_{k \to \infty} \frac{k^2\sqrt{3}}{2h^2}  \cdot \log_2\left(1+ \frac{1} {\mu\sigma^2 h^{\alpha}  + 3k(k+1) + 6 k^{\alpha} \zeta\left(\alpha-1, 1+ k\right)}\right).
\end{multline}
Next, by applying Watson's Lemma to the integral representation of the Hurwitz Zeta function we obtain the asymptotic expansion of $\zeta(\alpha-1, 1+ k)$ for $k \to \infty$.

\begin{multline} \zeta(\alpha-1, 1 + k) = \frac{1}{\Gamma(\alpha-1)}  \int_{0}^{\infty} \frac{x^{\alpha - 2} e^{-(1+k)x}}{1-e^{-x}}dx \sim \frac{\Gamma(\alpha - 2)}{\Gamma(\alpha-1)}(1+k)^{-\alpha + 2} + \frac{1}{2} (1+k)^{-\alpha+1} + O(k^{-\alpha}).\end{multline}

By substituting $\zeta(\alpha-1, 1 + k)$ for its asymptotic expansion in (\ref{proposition_2d_proof_eq_2}) and by applying the L'Hopital rule $ \lim_{x\to c} \frac{f(x)}{g(x)} = \lim_{x\to c} \frac{f'(x)}{g'(x)} $ we provide the lower bound for $l_1(r)$ as

\begin{align} 
&\lim_{\lambda\to \infty} \lambda\,\hat{\tau}^{\{2,1\}}_0(\lambda, \alpha) > \lim_{k \to \infty} \frac{k^2\sqrt{3}}{2h^2}\log_2\Bigg[1+ \bigg(\mu\sigma^2 h^{\alpha} + 3k(k+1) + 6 k^{\alpha} \Big(\frac{\Gamma(\alpha - 2)}{\Gamma(\alpha-1)(1+k)^{\alpha - 2}}& \nonumber \\ 
& + \frac{1}{2(1+k)^{\alpha - 1}}\Big)\bigg)^{-1}\Bigg] = \lim_{t \to 0^+} \frac{\sqrt{3}}{2t^2h^2}\log_2\Bigg[1+ t^2 \cdot \bigg(\mu\sigma^2 h^{\alpha}t^2  + 3(1+t)  
+ 6 \Big(\frac{\Gamma(\alpha - 2)}{\Gamma(\alpha-1)(1+t)^{\alpha - 2}} & \nonumber \\
& + \frac{t}{2(1+t)^{\alpha - 1}}\Big)\bigg)^{-1}\Bigg] =   \frac {2\sqrt{3}}{12\ln2\,h^2\left(1+2\,\frac{\Gamma(\alpha-2)}{\Gamma(\alpha-1)}\right)}.
\end{align}

From the above result it can be easily seen that the lower bound of ASE as $\lambda \to \infty$ for 2D regular network without fading under $l_1(r)$ is greater than zero and depends on $h$ and $\alpha$. Next, similar to 1D network, by using Proposition \ref{proposition_10a} and Lemma \ref{lemma_appendix} it can be shown that this result can be extended to 2D irregular network and 2D regular network with Rayleigh fading, thus concluding the proof for 2D network.
\end{proof}

\section{Proof of Proposition \ref{proposition_12}}
\label{proposition_12_proof}

The proof follows that of Theorem \ref{theorem_aase}. We start by substituting $l(r) = l_1(r)$ into (\ref{generic_1d_limit}) and rearranging the terms to obtain the following expression for the ASE limit for 1D regular network without fading

\begin{equation*} 
\lim_{\lambda \to \infty} \lambda\,\tau^{\{1,1\}} (\lambda, \alpha)  = 
\lim_{\Upsilon \to 0^+} \frac{1}{\Upsilon} \log_2 \Bigg(1 + \frac{h^{-\alpha}}{\mu\sigma^2 + 2\Upsilon^{-\alpha}\sum_{i=1}^{\infty}\left(i^{2} + \left(\frac{h}{d}\right)^2\right)^{-\alpha/2}} \Bigg).
\end{equation*}

Then, by substituting $k \to \frac{h}{\Upsilon}$ we get

\begin{equation}
\label{eq_prop12_1} 
\lim_{\lambda \to \infty} \lambda\,\tau^{\{1,1\}} (\lambda, \alpha)  = 
\lim_{k \to \infty} \frac{k}{h} \log_2 \left(1 + \frac{h^{-\alpha}}{\mu \sigma^2 + 2h^{-\alpha} k^\alpha \sum_{i=1}^{\infty}\left(i^{2} + k^2\right)^{-\alpha/2}} \right).
\end{equation}

By using (\ref{eqn_dbl_z}), given $n=\alpha/2$, and assuming $ \alpha = 2 $, we obtain the following expression 

\begin{figure*}[!t]
\normalsize
\setcounter{equation}{96}
\begin{multline}
\label{eqn_dbl_z}
\sum_{k=1}^{\infty} \frac{1}{(c^2+k^2)^{n}} =  \frac{\Gamma(n-\frac{1}{2})}{2 \sqrt{\pi}\,  \Gamma(n)\, c^{2n}} \left(\pi\,c\,\coth(\pi c)-1 \right) + \left( \frac{j}{2 c}\right)^n \sum_{k=0}^{n-2} \frac{(n)_k}{k!} \left(\frac{j }{2c}\right)^k \cdot \\ \Big( (-1)^n \zeta \left(n-k,1-jc\right) + (-1)^k \zeta \left(n-k, 1 + jc\right)\Big) 
\end{multline}
\hrulefill
\vspace*{4pt}
\end{figure*}

\setcounter{equation}{94}

\begin{equation} \sum_{i=1}^{\infty} \frac{1}{(i^2+k^2)} =  -\frac{1}{2 k^2} + \frac{\pi \coth(\pi k) }{2k}. \end{equation}

We use the above expression to simplify (\ref{eq_prop12_1}). By further applying the L'Hopital rule,  we obtain the ASE limit for $\alpha = 2$ as follows

\begin{multline}
\lim_{\lambda \to \infty} \lambda\,\tau^{\{1,1\}} (\lambda, \alpha)  = \lim_{k\to \infty} \frac{k}{h}\log_2 \left( 1+ \frac{h^{-2}} {\mu\sigma^2  + 2\, h^{-2}\, k^{2} \left( -\frac{1}{2 k^2} + \frac{\pi \coth(\pi k) }{2k}\right)}\right) 
= \frac{1}{\ln2\,\pi h}.
\end{multline}

Following a similar approach, and by using (\ref{eqn_dbl_z}) for  $n=\alpha/2$ we can prove ASE limits for $\alpha = 4$ and $\alpha = 6$.

Next, by using Proposition \ref{proposition_10a} and Lemma \ref{lemma_appendix} it can be shown that the results can be extended to 1D irregular network and 1D regular network with Rayleigh fading, thus concluding the proof.

\setcounter{equation}{97}

\section{Proof of Proposition \ref{proposition_2d_2}}
\label{proposition_2d_2_proof}

The proof follows that of Theorem \ref{theorem_aase}. We start by substituting $l(r) = l_1(r)$ into (\ref{generic_2d_limit}) and rearranging the terms to obtain the following limit expression 

\begin{equation}
\lim_{\lambda\to \infty} \lambda \hat{\tau}^{\{2, 1\}}_0(\lambda, \alpha)
 = \lim_{\lambda\to \infty} \lambda \log_2\left(1+ \frac{l_1(\|b_0\|)}{\mu\sigma^2 + \sum_{\Phi^{HEX} \setminus b_o} l_1(\|b_i\|)}\right).
\end{equation}

Next, similar to the proof of Theorem \ref{theorem_aase}, we use the expressions for lower and upper bounds of cumulated interference in hexagonal (triangular) lattice given that a typical user and its serving BS are located at the origin

\begin{equation} \sum_{k=1}^{\infty} 6k \cdot l_1(\Upsilon k)< I_0 < 6\,l_1(\Upsilon) + \sum_{k=2}^{\infty} 6k \cdot l_1(\Upsilon \sqrt{3}/2\, k)\end{equation}

Assuming $\alpha = 4$ and using the following formula for infinite series

\begin{equation} \sum_{k=n+1}^{\infty} \frac{k}{(k^2+c^2)^2} =  \frac{j}{4 c} (\psi^{(1)}(n+jc) - \psi^{(1)}(n-jc)) \end{equation}
we obtain inequalities presented below

\begin{equation} I_0 > 6\Upsilon^{-4} \frac{j}{4} \frac{\Upsilon}{h} (\psi^{(1)}(j\, h/\Upsilon) - \psi^{(1)}(-j\,h/\Upsilon))\qquad\end{equation}

\begin{multline}
I_0 < 6\Upsilon^{-4}\Bigg[ \bigg(1+\Big(\frac{h}{\Upsilon}\Big)^2\bigg)^{-2} + (\sqrt{3}/2)^{-3} \frac{j}{4} \cdot \frac{\Upsilon}{h} \bigg(\psi^{(1)}\Big(1+j\, \frac{2h}{\sqrt{3}\Upsilon}\Big) - \psi^{(1)}\Big(1-j\,\frac{2h}{\sqrt{3}\Upsilon}\Big)\bigg)\Bigg]
\end{multline}
where $\psi^{(1)}(x)$ is the Polygamma function of the first order, and $j=\sqrt{-1}$.

By taking the lower bound of the cumulated interference and substituting $k=\frac{h}{\Upsilon}$ we obtain the upper bound of $\lim_{\lambda \to \infty} \lambda \tau^{\{2, 2\}}_0(\lambda, 4)$ in the form presented below

\begin{multline}
\lim_{\lambda\to \infty} \lambda\,\tau^{\{2,1\}}_0(\lambda, 4) = \lim_{k\to \infty} \frac{2k^2}{h^2\sqrt{3}}  \cdot\log_2\left(1+ \frac{1} {\mu\sigma^2h^{4}  + 6 k^{3} \frac{j }{4} (\psi^{(1)}(jk) - \psi^{(1)}(-jk))}\right).
\end{multline}

Next, by applying Watson's Lemma to the integral representation of the Polygamma function of the first order we obtain the asymptotic expansion of $\psi^{(1)}(jk)$ in the following form
\begin{equation} \psi^{(1)}(jk) = \int_{0}^{\infty} \frac{x e^{jk x}}{1-e^{-x}}dx \sim (jk)^{-1} + \frac{1}{2} (jk)^{-2} + O((jk)^{-3}).\end{equation}

By substituting $\psi^{(1)}(jk)$ and $\psi^{(1)}(-jk)$ with their asymptotic expansions and then applying the L'Hopital rule, we obtain the upper bound for 2D regular network without fading.

\begin{multline}
\lim_{\lambda\to \infty} \lambda\,\tau^{\{2,1\}}_0(\lambda, 4) = \lim_{k\to \infty} \frac{2k^2}{h^2\sqrt{3}}
\cdot\log_2\left(1+ \frac{1} {\mu\sigma^2h^{4}  + 6 k^{3} \frac{j }{4} ((jk)^{-1} - (-jk)^{-1})}\right) \\
= \lim_{t\to 0^+} \frac{2}{t^2h^2\sqrt{3}} \cdot\log_2\left(1+ \frac{t^2} {\mu\sigma^2h^{4}t^2  + 3}\right) 
= \frac{2\sqrt{3}}{9\ln2\,h^2}
\end{multline}

Following a similar approach the lower bound of $\lim_{\lambda \to \infty} \lambda \tau^{\{2, 2\}}_0(\lambda, 4)$ can also be obtained. 

Next, using Proposition \ref{proposition_10a} and Lemma \ref{lemma_appendix} it can be shown that the results can be extended to 2D irregular networks and 2D regular networks with Rayleigh fading, thus concluding the proof.




\ifCLASSOPTIONcaptionsoff
  \newpage
\fi



%
%
%

\addcontentsline{toc}{chapter}{Bibliography} 
\bibliographystyle{IEEEtran}
\bibliography{new_jrnl_antenna_final}

\begin{thebibliography}{10}
\providecommand{\url}[1]{#1}
\csname url@samestyle\endcsname
\providecommand{\newblock}{\relax}
\providecommand{\bibinfo}[2]{#2}
\providecommand{\BIBentrySTDinterwordspacing}{\spaceskip=0pt\relax}
\providecommand{\BIBentryALTinterwordstretchfactor}{4}
\providecommand{\BIBentryALTinterwordspacing}{\spaceskip=\fontdimen2\font plus
\BIBentryALTinterwordstretchfactor\fontdimen3\font minus
  \fontdimen4\font\relax}
\providecommand{\BIBforeignlanguage}[2]{{%
\expandafter\ifx\csname l@#1\endcsname\relax
\typeout{** WARNING: IEEEtran.bst: No hyphenation pattern has been}%
\typeout{** loaded for the language `#1'. Using the pattern for}%
\typeout{** the default language instead.}%
\else
\language=\csname l@#1\endcsname
\fi
#2}}
\providecommand{\BIBdecl}{\relax}
\BIBdecl

\bibitem{6476879}
J.~Zander and P.~Mähönen, ``Riding the data tsunami in the cloud: myths and
  challenges in future wireless access,'' \emph{Communications Magazine, IEEE},
  vol.~51, no.~3, pp. 145--151, March 2013.

\bibitem{6845056}
E.~Hossain, M.~Rasti, H.~Tabassum, and A.~Abdelnasser, ``Evolution toward 5g
  multi-tier cellular wireless networks: An interference management
  perspective,'' \emph{Wireless Communications, IEEE}, vol.~21, no.~3, pp.
  118--127, June 2014.

\bibitem{ding2017performance2}
M.~Ding and D.~L{\'o}pez-P{\'e}rez, ``On the performance of practical
  ultra-dense networks: The major and minor factors,'' \emph{arXiv preprint
  arXiv:1701.07964}, 2017.

\bibitem{andrews2011tractable}
J.~G. Andrews, F.~Baccelli, and R.~K. Ganti, ``A tractable approach to coverage
  and rate in cellular networks,'' \emph{IEEE Transactions on Communications},
  vol.~59, no.~11, pp. 3122--3134, 2011.

\bibitem{zhang2015downlink}
X.~Zhang and J.~G. Andrews, ``Downlink cellular network analysis with
  multi-slope path loss models,'' \emph{IEEE Transactions on Communications},
  vol.~63, no.~5, pp. 1881--1894, 2015.

\bibitem{dhillon2012modeling}
H.~S. Dhillon, R.~K. Ganti, F.~Baccelli, and J.~G. Andrews, ``Modeling and
  analysis of k-tier downlink heterogeneous cellular networks,'' \emph{IEEE
  Journal on Selected Areas in Communications}, vol.~30, no.~3, pp. 550--560,
  2012.

\bibitem{gupta2015sinr}
A.~K. Gupta, X.~Zhang, and J.~G. Andrews, ``Sinr and throughput scaling in
  ultradense urban cellular networks,'' \emph{IEEE Wireless Communications
  Letters}, vol.~4, no.~6, pp. 605--608, 2015.

\bibitem{andrews2015we}
J.~G. Andrews, X.~Zhang, G.~D. Durgin, and A.~K. Gupta, ``Are we approaching
  the fundamental limits of wireless network densification?'' \emph{arXiv
  preprint arXiv:1512.00413}, 2015.

\bibitem{nguyen2016coverage}
V.~M. Nguyen and M.~Kountouris, ``Coverage and capacity scaling laws in
  downlink ultra-dense cellular networks,'' in \emph{2016 IEEE International
  Conference on Communications (ICC)}.\hskip 1em plus 0.5em minus 0.4em\relax
  IEEE, 2016, pp. 1--7.

\bibitem{ding2016performance}
M.~Ding, P.~Wang, D.~L{\'o}pez-P{\'e}rez, G.~Mao, and Z.~Lin, ``Performance
  impact of los and nlos transmissions in dense cellular networks,'' \emph{IEEE
  Transactions on Wireless Communications}, vol.~15, no.~3, pp. 2365--2380,
  2016.

\bibitem{gupta2015potential}
A.~K. Gupta, X.~Zhang, and J.~G. Andrews, ``Potential throughput in 3d
  ultradense cellular networks,'' in \emph{2015 49th Asilomar Conference on
  Signals, Systems and Computers}.\hskip 1em plus 0.5em minus 0.4em\relax IEEE,
  2015, pp. 1026--1030.

\bibitem{7842150}
M.~Ding and D.~L. Perez, ``Please lower small cell antenna heights in 5g,'' in
  \emph{2016 IEEE Global Communications Conference (GLOBECOM)}, Dec 2016, pp.
  1--6.

\bibitem{ding2017performance}
M.~Ding and D.~Lopez-Perez, ``Performance impact of base station antenna
  heights in dense cellular networks,'' \emph{arXiv preprint arXiv:1704.05125},
  2017.

\bibitem{7962695}
M.~Filo, C.~H. Foh, S.~Vahid, and R.~Tafazolli, ``Performance impact of antenna
  height in ultra-dense cellular networks,'' in \emph{2017 IEEE International
  Conference on Communications Workshops (ICC Workshops)}, May 2017, pp.
  429--434.

\bibitem{chen2012small}
C.~S. Chen, V.~M. Nguyen, and L.~Thomas, ``On small cell network deployment: A
  comparative study of random and grid topologies,'' in \emph{Vehicular
  Technology Conference (VTC Fall), 2012 IEEE}.\hskip 1em plus 0.5em minus
  0.4em\relax IEEE, 2012, pp. 1--5.

\bibitem{baccelli2009stochastic}
F.~Baccelli and B.~Blaszczyszyn, \emph{Stochastic geometry and wireless
  networks}.\hskip 1em plus 0.5em minus 0.4em\relax Now Publishers Inc, 2009,
  vol.~1.

\bibitem{haenggi2009interference}
M.~Haenggi and R.~K. Ganti, \emph{Interference in large wireless
  networks}.\hskip 1em plus 0.5em minus 0.4em\relax Now Publishers Inc, 2009.

\bibitem{li2001capacity}
L.~Li and A.~J. Goldsmith, ``Capacity and optimal resource allocation for
  fading broadcast channels. i. ergodic capacity,'' \emph{IEEE Transactions on
  Information Theory}, vol.~47, no.~3, pp. 1083--1102, 2001.

\bibitem{zhang2014performance}
X.~Zhang and M.~Haenggi, ``The performance of successive interference
  cancellation in random wireless networks,'' \emph{IEEE Transactions on
  Information Theory}, vol.~60, no.~10, pp. 6368--6388, 2014.

\bibitem{madhusudhanan2012downlink}
P.~Madhusudhanan, J.~G. Restrepo, Y.~Liu, and T.~X. Brown, ``Downlink coverage
  analysis in a heterogeneous cellular network,'' in \emph{Global
  Communications Conference (GLOBECOM), 2012 IEEE}.\hskip 1em plus 0.5em minus
  0.4em\relax IEEE, 2012, pp. 4170--4175.

\bibitem{baccelli2015correlated}
F.~Baccelli and X.~Zhang, ``A correlated shadowing model for urban wireless
  networks,'' in \emph{2015 IEEE Conference on Computer Communications
  (INFOCOM)}.\hskip 1em plus 0.5em minus 0.4em\relax IEEE, 2015, pp. 801--809.

\end{thebibliography}

%




\end{document}